\documentclass[journal]{IEEEtran}

\usepackage{xr}
\usepackage{forest}
\usepackage{signalpreamble}
\usepackage{amsthm}

\newtheorem{problem}{Problem}
\newtheorem{subproblem}{Problem}[Problem]

\myexternaldocument{supplementary}

\newcommand{\boldell}{{\bm{\ell}}}

\DeclareMathOperator{\VT}{VT}
\DeclareMathOperator{\SF}{SF}

\newacronym{GSP}{GSP}{graph signal processing}
\newacronym{GGSP}{GGSP}{generalized graph signal processing}
\newacronym{PSWF}{PSWF}{prolate spheroidal wave function}
\newacronym{GFT}{GFT}{graph Fourier transform}
\newacronym{FT}{FT}{Fourier transform}
\newacronym{WLOG}{WLOG}{without loss of generality}
\newacronym{JFT}{JFT}{joint Fourier transform}
\newacronym{STVFT}{STVFT}{short time-vertex Fourier transform}
\newacronym{STVWT}{STVWT}{spectral time-vertex wavelet transform}
\newacronym{NEGUP}{NEGUP}{Naive extension of graph uncertainty principle}

\begin{document}

\title{Uncertainty Principle for Vertex-Time Graph Signal Processing}
\author{Yanan Zhao, Xingchao Jian, Feng Ji, Wee~Peng~Tay,~\IEEEmembership{Senior~Member,~IEEE}, Antonio Ortega,~\IEEEmembership{Fellow,~IEEE}%
\thanks{Yanan Zhao, Xingchao Jian, Feng Ji, and Wee Peng Tay are with the School of Electrical and Electronic Engineering, Nanyang Technological University, Singapore.}
\thanks{Antonio Ortega is with University of Southern California, Los Angeles, USA.}
}

\markboth{submitted to IEEE TRANSACTIONS ON SIGNAL PROCESSING}%
{How to Use the IEEEtran \LaTeX \ Template}


\maketitle 



\begin{abstract}

We present an uncertainty principle for graph signals in the vertex-time domain, unifying the classical time-frequency and graph uncertainty principles within a single framework. By defining vertex-time and spectral-frequency spreads, we quantify signal localization across these domains. Our framework identifies a class of signals that achieve maximum concentration in both the spatial and temporal domains. These signals serve as fundamental atoms for a new vertex-time dictionary, enhancing signal reconstruction under practical constraints, such as intermittent data commonly encountered in sensor and social networks. Furthermore, we introduce a novel graph topology inference method leveraging the uncertainty principle. Numerical experiments on synthetic and real datasets validate the effectiveness of our approach, demonstrating improved reconstruction accuracy,  greater robustness to noise, and enhanced graph learning performance compared to existing methods.

\end{abstract}


\section{Introduction}

The uncertainty principle for signal processing, first developed by Slepian, Landau, and Pollack \cite{Slepian1983,Laudau1980}, describes the trade-off between a signal’s time and frequency localization. It states that no signal can be simultaneously arbitrarily well localized in both domains. Prolate spheroidal wave functions (PSWFs) \cite{Slepian1961,Pollak1961,Laudau1962,Slepian1978} embody this balance by maximizing a signal’s energy within a finite time interval while remaining bandlimited to a specified frequency range, making them an effective basis for representing or processing signals under simultaneous time-frequency constraints.

Similarly, uncertainty principles for \gls{GSP} have been developed in \cite{Ameya:J2013,Pasdeloup2015,Benedetto2015,Koprowski2016} and more recently in \cite{Oguzhan2017,Erb2021}, capturing the trade-off between localization in the vertex and (graph) spectral domains. 
Early formulations of graph uncertainty principles, such as \cite{Ameya:J2013}, define signal spread using distance-based metrics over the graph, measuring localization in the vertex and spectral domains. However, this approach depends on a specific choice of graph distance, such as the geodesic distance \cite{Royle2001}, resistance distance \cite{Klein1993}, or diffusion distance \cite{Coifman2006}, 
which introduces inherent limitations.
To address this, \cite{Tsitsvero2016} introduced a distance-independent definition of spread inspired by the work of Slepian, Landau, and Pollack \cite{Slepian1961,Pollak1961,Slepian1983} on prolate spheroidal wave functions. 
Their approach characterizes vertex and spectral spreads without relying on specific distances, defining a principle that measures energy concentration across vertex and spectral subsets. Notably, \cite{Tsitsvero2016} observes that, unlike the uncertainty principle for time-frequency localization, certain graph signals can be perfectly localized in both domains. By defining a feasible boundary for the vertex and spectral spreads, a class of signals maximally concentrated in both the vertex and spectral domains can be identified. 
Further extensions of the uncertainty principle have been made in the context of the graph linear canonical transform (GLCT) domain \cite{ZHANG2025}, leading to a sampling theory for the GLCT and the derivation of optimal sampling operators based on the uncertainty principle.

Traditional \gls{GSP} focuses mainly on analyzing information within the vertex domain. 
In some applications, however, each vertex is associated with a discrete or continuous \textbf{time} signal. To address such scenarios, the works \cite{Perraudin2015,Loukas2016,JiTay:J19,Zhang2021} have extended classical GSP to handle vertex-time signals. 
This includes tools such as the vertex-time Fourier transform, vertex-time filters, joint harmonic analysis, and vertex-time signal recovery. 
To enable localized representations in the \emph{vertex-time domain}, redundant vertex-time dictionaries \cite{Grassi2016,Grassi2018} have also been developed, including the \gls{STVFT} and the \gls{STVWT}.   

Despite these advances, existing methods struggle to handle sparse, irregular, or constrained observations, as commonly encountered in real-world scenarios such as sensor networks. 
Some specific challenges are: 
(i) \emph{temporal irregularity}, where sensors record at uneven or intermittent intervals due to limited battery, event-triggered sampling, or communication delays; (ii) \emph{spatial sparsity}, where only a subset of sensors report at a given time because others are shut down, fail, or undergo maintenance; 
(iii) \emph{joint spatio–temporal irregularity}, where the set of reporting sensors changes over time, leaving irregular coverage of the vertex–time domain. 
Conventional approaches like STVFT and STVWT rely on fixed-time windowing or wavelet bases that assume regularly spaced coverage of the vertex–time domain and thus lack the adaptability to capture signal energy concentrated over irregular or data-dependent subsets of vertices and time intervals. These limitations are especially problematic in sensor networks, where solar-powered devices may only operate during daylight or harsh environments may force intermittent operation, all of which reduce data availability and hinder reconstruction. 

In this work, we refer to the ``graph frequency'' domain as the \emph{spectral} domain, and the ``time frequency'' domain as simply the \emph{frequency} domain. 
Addressing the above challenges requires signal processing and representation methods that account for the localization trade-offs between the vertex-time and spectral-frequency domains. 
While uncertainty principles for the vertex domain versus the spectral domain in \gls{GSP} have been established, none address signals defined over irregular vertex–time domains. To close this gap, our objective is to establish an uncertainty principle that rigorously characterizes the localization trade-offs in the vertex–time domain versus the spectral-frequency domain, which in turn provides the theoretical foundation for designing energy-concentrated dictionaries and graph learning.

\begin{figure}[!htbp]
\centering
\vspace{-5mm}
\begin{subfigure}[b]{0.49\columnwidth}
\centering
\includegraphics[width=\columnwidth, trim={3.8cm 9cm 3.8cm 9cm}, clip]{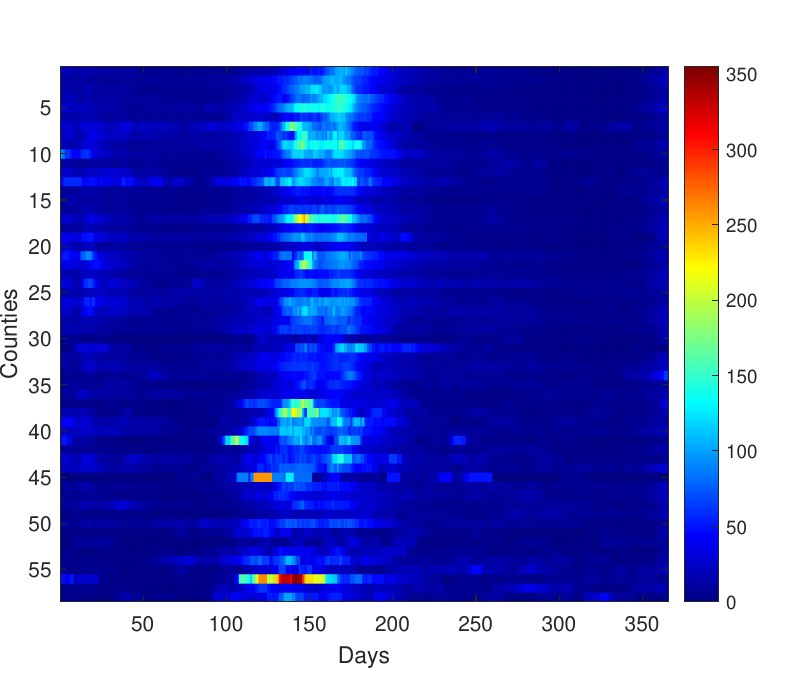}
\vspace{-6mm}
\caption{Jul 29, 2020--Jul 30, 2021}
\end{subfigure}
\hfill
\begin{subfigure}[b]{0.49\columnwidth}
\centering
\includegraphics[width=\columnwidth, trim={3.8cm 9cm 3.8cm 9cm}, clip]{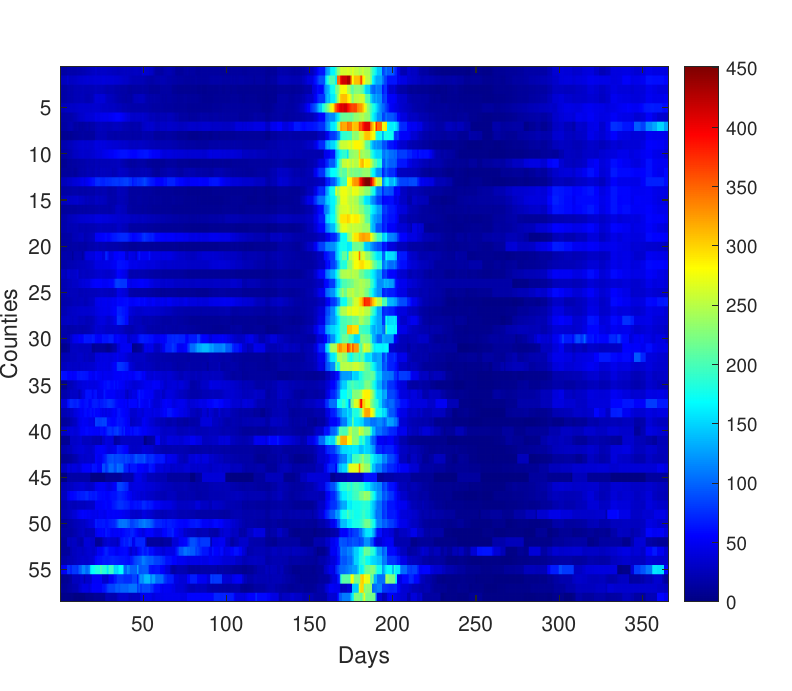}
\vspace{-6mm}
\caption{Jul 31, 2021--Aug 1, 2022}
\end{subfigure}
\vspace{-7mm}
\caption{The number of reported cases across counties and days from July 29, 2020, to August 1, 2022, for the Covid-19 dataset.}
\label{fig:heat_map_covid_19}
\end{figure}

To illustrate the need for energy-concentrated representations, consider the Covid-19 dataset\footnote{\url{https://github.com/TorchSpatiotemporal/tsl}}, which records reported cases across California’s 58 counties from July 29, 2020, to August 1, 2022. \Cref{fig:heat_map_covid_19} shows that case counts are highly localized, with energy concentrated in specific counties (e.g., 2, 5, 14, 15, 57) and during particular time intervals (e.g., days 100 to 200). In addition, the dataset contains days with zero reported cases in some counties, which often reflect delayed or missing reporting rather than a true absence of cases. These patterns demonstrate that the signal energy is not uniformly distributed in time or across vertices, but instead occupies a limited subset of the vertex-time domain. 
This highlights the importance of constructing vertex-time dictionaries that can efficiently capture such localized energy concentrations. 
However, existing approaches such as STVFT and STVWT rely on fixed bases that are not tailored to irregular or data-driven supports, limiting their ability to represent signals concentrated on specific subsets in the joint vertex-time domain. 
In this work, we address this limitation by first defining the uncertainty principle in the vertex-time domain and then leveraging it to design basis functions that maximize energy concentration over selected subsets of vertices and time intervals, enabling efficient signal representation and reconstruction under practical, constrained observation scenarios.

In this paper, we focus on energy-based localization over the vertex-time domain, while acknowledging the existence of alternative localization metrics \cite{Stankovic2017,Shuman2017,Sandryhaila2014}. 
Unlike prior works such as \cite{Perraudin2015,Loukas2016,Zhang2021} that operate in the discrete-time domain with signals defined at fixed time steps, our framework adopts the continuous-time domain introduced in \cite{JiTay:J19}. 
Discrete-time formulations face inherent limitations when node sampling is irregular. 
Using a common discrete sampling rate enforces temporal alignment but requires the rate to match the most densely sampled node, resulting in a high global rate and many missing entries for sparsely observed nodes. Conversely, assigning independent sampling rates to each node better reflects their observation patterns, but this approach leads to temporal misalignment and complicates cross-node comparisons.
A continuous-time formulation overcomes these issues by defining signals on a shared temporal axis and interpreting discrete observations as samples from an underlying continuous process. This provides a unified representation that naturally accommodates asynchronous and heterogeneous sampling without resampling, while maintaining consistent alignment across nodes and enabling a precise characterization of energy concentration over irregular subsets of the vertex–time domain. 
Building on this perspective, we develop an uncertainty principle to understand the trade-offs between signal localization in the vertex-time and spectral-frequency domains. This principle underpins two main applications: \emph{(i)} the design of basis functions that are maximally energy-concentrated over specific vertex–time subsets, and \emph{(ii)} a graph learning approach that maximizes spectral–frequency energy concentration while respecting vertex–spectral localization trade-offs, which existing methods \cite{stefania2019,dong2016} typically ignore. 

Instead of naively extending the vertex-domain framework in \cite{Tsitsvero2016}, which considers signal localization only in the vertex domain, by simply including the entire time interval, our new framework captures the localized behavior of signals across both vertices and time intervals by focusing on specific subspaces of functions defined on the joint vertex-time domain. 
This approach leads to a better understanding of how signals evolve over space and time, and to a more accurate and efficient representation of signals. 
The main contributions of this paper are summarized as follows:  
\begin{enumerate} [1)]
\item We formulate a \textit{generalized uncertainty principle in the vertex-time domain} that unifies and extends existing graph \cite{Tsitsvero2016} and continuous-time formulations \cite{Slepian1983} through a subspace-based localization measure, encompassing the classical subset-based notions \cite{Tsitsvero2016,Laudau1980} as special cases.
\item We \textit{identify a class of vertex-time graph signals that achieve maximal concentration in both the vertex-time and spectral-frequency domains}. These signals serve as fundamental atoms in a vertex-time dictionary for signal reconstruction. We derive their localization properties and outline dictionary construction and optimization for improved reconstruction accuracy.
\item We propose a \textit{new graph learning approach that maximizes energy concentration in the spectral-frequency domain}. This method captures the intrinsic signal structure while adhering to the vertex-spectral localization trade-off dictated by the uncertainty principle.
\item We \textit{evaluate our vertex-time dictionary learning and graph learning methods on both synthetic and real datasets}. Results demonstrate improved reconstruction accuracy, achieving an average gain of 4 dB, and superior noise robustness, maintaining stable performance even under noise levels as low as 0 dB, compared to existing dictionary learning methods. Additionally, our method improves graph inference precision by 4.7\% over existing graph learning approaches.
\end{enumerate}

This paper extends our earlier work in \cite{zhao2024UPICASSP}, where we first introduced a generalized uncertainty principle that unifies existing principles for graph and continuous-time signals. 
In this paper, we build upon that foundation by providing a detailed characterization of localization trade-offs in the vertex–time and spectral–frequency domains. Leveraging these insights, we propose a new energy-concentrated vertex–time dictionary for signal reconstruction and introduce a new graph learning approach that captures intrinsic signal structures while explicitly respecting the localization trade-offs dictated by the uncertainty principle. These extensions address the main limitations of our earlier work, namely the lack of a systematic study of concrete localization conditions (e.g., when perfect vertex and time localization is achievable) and the absence of concrete applications of the established principle.        

The remainder of this paper is organized as follows. \cref{sec.UP} develops the uncertainty principle for vertex-time graph signals. Compared to our earlier work \cite{zhao2024UPICASSP}, \cref{sec.joint_dic} expands the discussion by proposing a new vertex-time dictionary learning method based on this principle for signal reconstruction, while \cref{sec.gl} introduces a new graph topology inference method that was not included in \cite{zhao2024UPICASSP}. Numerical results are presented in \cref{sec.exp}, followed by concluding remarks in \cref{sec.conclu}.

\emph{Notations.} Scalars and scalar-valued functions are denoted by lowercase letters (e.g., $x$), vectors and vector-valued functions by bold lowercase (e.g., $\bx$), and matrices by bold uppercase (e.g., $\bA$). The sets of real, complex, and integer numbers are denoted by $\Real$, $\Complex$, and $\bbZ$, respectively. The space of square-integrable functions on a domain $\Omega$ is $L^{2}(\Omega)$, and $\|\cdot\|$ denotes its norm when unambiguous. The symbols $\|\cdot\|_1$, $\|\cdot\|_2$, and $\|\cdot\|_F$ denote the $L^1$, $L^2$, and Frobenius norms. The Kronecker delta is $\delta_{ij}$, where $i$ and $j$ denote integer indices. The tensor product is $\otimes$, $\diag(\bv)$ is the diagonal matrix with diagonal $\bv$, and transpose is denoted by $(\cdot)\T$.

\section{Uncertainty Principle for Vertex-time Graph Signals}
\label{sec.UP}

In this section, we establish an uncertainty principle for vertex-time graph signals. We adopt and extend the definition of spread introduced in \cite{Tsitsvero2016} to quantify the localization of vertex-time graph signals in the vertex-time and spectral-frequency domains.

\subsection{Fourier transforms}

Throughout this paper, we consider undirected graphs $\calG=(\calV,\bA)$ consisting of a set of $N$ vertices $\calV=\set{1,2,\dots,N}$, along with a set of weighted edges with adjacency weight matrix $\bA=[a_{ij}]_{i,j\in \mathcal{V}}$, such that $a_{ij}>0$, if vertex $i$ and $j$ are connected by an edge, or $a_{ij}=0$, otherwise. We assume that $\mathcal{G}$ is a connected weighted graph with no self-loops. 
The graph Laplacian matrix is given by $\bL=\bD-\bA$, where $\bD$ is the degree matrix. We take $\bL$ to be the graph shift operator (GSO). 
The \gls{GFT} of a graph signal $\bx\in\bbR^N$ \cite{ortega2022introgsp,Shuman2013,Pesenson2010,Sandryhaila2013,Pesenson2008,Rabbat2012} is given by
\begin{align}
\label{eq.graph_Fourier_transform}
\widehat{\bx}=\bU{\T}\bx
\end{align}
where $\bU$ is the unitary matrix whose columns are the Laplacian eigenvectors. 

A vertex-time graph signal $f$ is a map $\calV\times \calT \rightarrow \bbC$, defined by $(v,t)\mapsto f(v, t)$, where $\calV$ and $\calT=\Real$ represent the \emph{vertex} and \emph{time} domains, respectively \cite{JiTay:J19,Grassi2018,Loukas2016}. We denote the space of vertex-time graph signals as $L^{2}\left(\calV\times\calT\right)$. For $f\in L^2(\calV\times \calT)$, the joint Fourier transform (JFT) is defined as \cite{JiTay:J19}: 
\begin{align}
\label{eq.joint_Fourier_transform}
\mathcal{F}_{f}(\lambda_k,\omega)= \langle f, \bu_{k}\otimes e^{\iu\omega t}\rangle, \;\; \text{for}\;\;  k=1,\dots,N, 
\end{align}
where $\lambda_k$ and $\bu_{k}$ are the $k$-th eigenvalue and eigenvector of $\bL$, respectively, $\otimes$ is the tensor product, and $\iu=\sqrt{-1}$. 
The JFT maps the signal $f$ in $L^2(\calV\times\calT)$ to a signal in $L^2(\Lambda\times\Omega)$, where $\Lambda = \set{\lambda_1,\dots,\lambda_N}$ and $\Omega=\bbR$ are referred to as the \emph{spectral} and \emph{frequency} domains, respectively. The JFT generalizes both the \gls{GFT} and the classical \gls{FT}. 
Specifically, the JFT reduces to the \gls{GFT}, which maps a graph signal from $\calV$ to $\Lambda$, when $f$ is evaluated at a fixed time instance, and to the \gls{FT} mapping a temporal signal from $\calT$ to $\Omega$ when $\calV$ is a singleton. Readers are referred to \cite{JiTay:J19} for more details on \gls{GGSP}.

\subsection{Vertex-time and spectral-frequency spreads}

In this work, we adopt a general concept of localization for square-integrable vertex-time graph signals based on energy concentration over subspaces of $L^{2}(\calV\times \calT)$. This approach extends the framework introduced in \cite{Tsitsvero2016} for graph signals.

\begin{Definition}[Vertex-time projection]\label{def.vertex-time_projection}
Let $\calH_{\VT} \subseteq  L^{2}(\calV\times \calT)$ be a closed subspace. The vertex-time projection operation $\Pi_{\VT}$ is the orthogonal projection onto $\calH_{\VT}$, where $\Pi_{\VT} f \in \calH_{\VT}$ for any vertex-time graph signal $f\in L^{2}(\calV\times \calT)$. 
\end{Definition}

The operator $\Pi_{\VT}$ is Hermitian and idempotent, i.e., $(\Pi_{\VT})^2=\Pi_{\VT}$. The following are special cases following  \cref{def.vertex-time_projection}:
\begin{enumerate}[i)]
\item Define $\calH_{\VT}: = \set{ f\in L^{2}(\calV\times\calT)\given f(v,t)=0 ~\text{for}~(v,t) \notin \calS}$ for a vertex–time support subset $\calS \subseteq \calV \times \calT$. The projection reduces to the classical vertex-time limiting operator:
\begin{align}\label{eq.joint_VT_opt}
\Pi_{\VT} f = \Pi_{\calS} f:=
\begin{cases}
f(v,t) &\ \text{if}\ (v,t)\in~\calS,\\
0 &\ \text{otherwise}.
\end{cases}
\end{align}
Specific examples include:
\begin{enumerate}[label=\alph*)]
    \item \emph{Vertex-limiting projection:} When the signal is defined over the vertex set $\calV$ at a fixed time $t_{0}$, we have $f\in L^{2}(\calV \times \{t_{0}\})\cong L^{2}(\calV)$, which naturally reduces to a graph signal $f(v)=f(v,t_0) \in \Complex$ for each $v\in \calV$. In this case, a subset $\calS=\calV'\times \set{t_{0}} \subset \calV\times\calT$ corresponds to restricting the signal to a vertex subset $\calV'$, and the operator $\Pi_{\VT}$ reduces to the vertex-limiting operator $\Pi_{\calV'}$ \cite{Tsitsvero2016}.
    \item \emph{Time-limiting projection:} For a trivial graph $\calG$ with a single vertex $v_{0}$, $f\in L^{2}(\set{v_{0}}\times \calT)\cong  L^2(\calT)$ is a continuous-time signal on $\calT=\bbR$. Let $\calT' = [t_{c} - \ell/2, t_{c} + \ell/2]$ be an interval within $\calT$, where $t_{c}$ is the center time point and $\ell$ is the length of $\calT'$. In this case, $\calS=\set{v_{0}}\times \calT'$, and the operator $\Pi_{\VT}$ reduces to the time-limiting operator $\Pi_{\calT'}$ in \cite{Slepian1961}, which restricts the signal $f(t)=f(v_{0},t)$ to $\calT'$. 
\end{enumerate}

\item Define 
\[\small
\cal H_{\VT}
=\overline{\mathrm{span}}
\left\{
\begin{aligned}
& e^{-\tau_{v}\bL}\delta_{v_{0}}\otimes \frac{1}{\sqrt{4\pi\tau_{t}}}
  \exp\!\Big(-\frac{(t-t_{0})^{2}}{4\tau_{t}}\Big) \\
& \qquad :\; v_{0}\in \calV',\; t_{0}\in \calT'
\end{aligned}
\right\},
\]


where $e^{-\tau_{v}\bL}\delta_{v_{0}}$ is the graph heat atom centered at vertex $v_{0}$, and $\frac{1}{\sqrt{4\pi\tau_{t}}}e^{-\frac{(t-t_{0})^{2}}{4\tau_{t}}}$ is the Gaussian heat kernel in the time domain centered at time $t_{0}$. 
The parameters $\tau_{v}>0$ and $\tau_{t}>0$ control the diffusion scales in the vertex and time domains, respectively, and allow one to model signals that are localized around specific vertices and time instances (cf.\ \cref{sec.gaussian_subspace}).

\end{enumerate}

\begin{Definition}[Spectral-frequency projection]\label{def.joint_fre_opt_subspace}
Let $\calU: L^{2}(\calV\times \calT) \rightarrow L^{2}(\Lambda\times \Omega)$ be the JFT in \cref{eq.joint_Fourier_transform}, and let $\calH_{\SF}\subseteq L^{2}(\Lambda\times \Omega)$ be a closed subspace. The spectral-frequency projection operator $\Pi_{\SF}$, for each vertex-time graph signal $f\in L^2(\calV\times\calT)$, is defined as follows:
\begin{align}\label{joint_fre_opt_subspace}
\Pi_{\SF} f = \calU^{-1}\Sigma_{\SF}\{\calU f\}
\end{align}
where $\Sigma_{\SF}$ denotes the orthogonal projection onto $\calH_{\SF}$. 
\end{Definition}

The operator $\Pi_{\SF}$ is Hermitian and idempotent, i.e., $\Pi_{\SF}^2=\Pi_{\SF}$. In practice, one often takes $\calH_{\SF}$ to be the subspace of functions supported on a measurable spectral-frequency region $\Sigma\subseteq \Lambda\times\Omega$, 
namely
$\calH_{\SF}=\set*{ g\in L^{2}(\Lambda\times\Omega) \given g(\lambda,\omega)=0\ \text{for}\ (\lambda,\omega)\notin\Sigma }$ for some measurable subset $\Sigma \subseteq \Lambda\times\Omega$.
In this case, $\Pi_{\SF}$ is the spectral-frequency limiting operator,\footnote{We omit $\calU$ in the notation $\Pi_{\Sigma}$ to avoid clutter, which is clear from the context.} defined by
\begin{align}\label{joint_fre_opt}
    \Pi_{\SF}f = \Pi_{\Sigma} := \calU^{-1}\Sigma\{\calU f\},
\end{align}
where
\begin{align*}
\Sigma\{\calU f\} :=
\begin{cases}
(\calU f)(\lambda,\omega) & \text{if } (\lambda,\omega)\in\Sigma,\\
0 & \text{otherwise},
\end{cases}
\end{align*}
i.e., $\Pi_{\SF}$ zeroes the JFT outside $\Sigma$ to form a bandlimited signal.
The operator $\Pi_{\SF}$ recovers familiar special cases:
\begin{enumerate}[i)]
\item For signals $f\in L^{2}(\calV\times\{t_0\})$ with $f(v)=f(v,t_0)\in\Complex$ for $v\in\calV$, the transform $\calU$ can be the \gls{GFT} in \cref{eq.graph_Fourier_transform}. Then $\Pi_{\SF}$ becomes the spectral-limiting projector $\Pi_{\Lambda'}=\bU\Sigma_{\Lambda'}\bU^{\ast}$, where $\Lambda'\subset\Lambda$ and $\Sigma_{\Lambda'}=\diag(\mathbf{1}_{\Lambda'})$, projecting $f$ onto the subspace spanned by eigenvectors indexed by $\Lambda'$ \cite{Tsitsvero2016}.

\item For a trivial graph consisting of a single vertex, $\calU$ can be the \gls{FT} $\calF$ mapping $L^{2}(\calT)$ to $L^{2}(\Omega)$. For a frequency band $\Omega'=[\omega_c-W,\omega_c+W]$, where $\omega_c$ denotes the center frequency and $W$ the bandwidth, $\Pi_{\SF}$ reduces to the classical frequency-limiting operator $\Pi_{\Omega'}$ \cite{Slepian1961,Pollak1961}, which produces the bandlimited function
\begin{align*}
\Pi_{\Omega'}f(t)=\frac{1}{2\pi}\int_{\omega_c-W}^{\omega_c+W}\calF_f(\omega)\,e^{\iu\omega t}\ud\omega.
\end{align*}
\end{enumerate}



Let $\calH_{\VT}\subseteq L^{2}(\calV\times\calT)$ be a closed subspace, with  $\calH_{\VT}^{\perp}$ denoting its orthogonal complement, and let $\overline{\Pi}_{\VT}$ denote the orthogonal projection onto $\calH_{\VT}^{\perp}$. 
Then $\overline{\Pi}_{\VT}=\calI-\Pi_{\VT}$, where $\calI$ is the identity on $L^{2}(\calV\times\calT)$. In particular $\overline{\Pi}_{\VT}$ is self‑adjoint and idempotent, i.e., $\overline{\Pi}_{\VT}^{\ast}=\overline{\Pi}_{\VT}$ and $\overline{\Pi}_{\VT}^{2}=\overline{\Pi}_{\VT}$. The same relations hold for the spectral–frequency projector: $\overline{\Pi}_{\SF}=\calI-\Pi_{\SF}$ and $\Pi_{\SF}+\overline{\Pi}_{\SF}=\calI$.

\begin{Definition}[Vertex-time spread and spectral-frequency spread] 
Given the vertex-time projection operator $\Pi_{\VT}$ and the spectral-frequency projection operator $\Pi_{\SF}$, the vertex-time spread and spectral-frequency spread of a vertex-time graph signal $f\in L^{2}(\calV\times \calT)$ are defined with respect to these subspaces as follows: 
\begin{align}\label{eq.joint_spread_subspace}
\alpha^2_{\VT} 
= \frac{\norm{\Pi_{\VT} f}_{\calH_{\VT}}^{2}}{\norm{f}_{L^{2}(\cal{V}\times \cal{T})}^{2}}
~~\text{and}~~
\beta^2_{\SF} 
= \frac{\norm{\Pi_{\SF} f}_{\calH_{\SF}}^{2}}{\norm{f}_{L^{2}(\cal{V}\times \cal{T})}^{2}}.
\end{align}
\end{Definition}

These quantities measure the fractions of signal energy contained in the vertex–time and spectral–frequency subspaces, respectively.

\subsection{Uncertainty principle for vertex-time graph signals}

Given a closed vertex-time subspace $\calH_{\VT}$ and a spectral-frequency subspace $\calH_{\SF}$, the uncertainty principle for vertex-time graph signals establishes a trade-off between $\alpha_{\VT}^2$ and $\beta_{\SF}^2$, as stated in the following theorem.

\begin{Theorem}\label{UP_GGSP}
For a given $\calH_{\VT}$ and $\calH_{\SF}$, the feasible region $\Theta$ of achievable pairs $(\alpha_{\VT}, \beta_{\SF})$ is given by
\begin{align}
\label{eq.Theorem_feasible region}
\begin{aligned}
&\Theta = \Bigg\{ (\alpha_{\VT},\beta_{\SF}) :\\   
&\cos^{-1} \alpha_{\VT} + \cos^{-1} \beta_{\SF} \geq \cos^{-1}  \sqrt{\lambda_{\max}(\Pi_{\SF}\Pi_{\VT}\Pi_{\SF})},\\
&\cos^{-1}  \sqrt{1 - \alpha_{\VT}^2}+\cos^{-1} \beta_{\SF} \geq \cos^{-1}  \sqrt{\lambda_{\max}(\Pi_{\SF}\overline{\Pi}_{\VT}\Pi_{\SF})},\\
&\cos^{-1}  \alpha_{\VT} + \cos^{-1}  \sqrt{1 - \beta_{\SF}^2} \geq \cos^{-1}  \sqrt{\lambda_{\max}(\overline{\Pi}_{\SF}\Pi_{\VT}\overline{\Pi}_{\SF})}\\
&\cos^{-1}  \sqrt{1 - \alpha_{\VT}^2} + \cos^{-1}  \sqrt{1 - \beta_{\SF}^2} \\
& \qquad\qquad \qquad \geq \cos^{-1}  \sqrt{\lambda_{\max}(\overline{\Pi}_{\SF}\overline{\Pi}_{\VT}\overline{\Pi}_{\SF})}
\Bigg\}
\end{aligned}
\end{align}
where $\lambda_{\max}(\cdot)$ denotes the largest eigenvalue of its operator argument. I.e., there exists a $f\in L^{2}(\calV\times \calT)$ with $\norm{\Pi_{\VT} f}=\alpha_{\VT}$ and $\norm{\Pi_{\SF} f}=\beta_{\SF}$ if and only if $(\alpha_{\VT}, \beta_{\SF}) \in \Theta$.
\end{Theorem}
\begin{IEEEproof}
The proof follows the same steps as the corresponding proof in \cite{Tsitsvero2016} and is provided for completeness in the Appendices.
\end{IEEEproof}

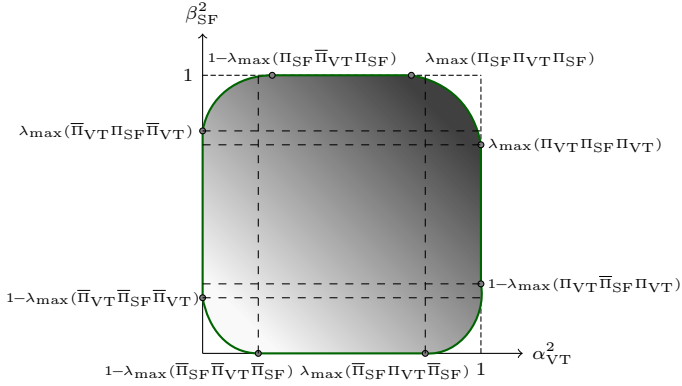
\begin{figure}
\centering
\begin{tikzpicture}[scale=3.7]

    \draw[dashed, dash pattern=on 2pt off 1pt] (0, 1) -- (1, 1);
    \draw[dashed, dash pattern=on 2pt off 1pt] (1, 0) -- (1, 1);

    \draw[->] (0, 0) -- (1.15, 0) node[right] {\scriptsize $\alpha_{\VT}^2$};
    \draw[->] (0, 0) -- (0, 1.15) node[above] {\scriptsize $\beta_{\SF}^2$};

    \shade[shading=axis, left color=gray!5, right color=gray!150, shading angle=135] 
    (0, 0.8) to[out=80, in=180] (0.25, 1) -- 
    (0.75, 1) to[out=-10, in=100] (1, 0.75) -- 
    (1, 0.25) to[out=-80, in=-5] (0.8, 0) -- 
    (0.2, 0) to[out=180, in=-80] (0, 0.2) -- 
    (0, 0.8);

    \draw[color={rgb,255:red,0; green,100; blue,0}, thick, smooth] (0, 0.8) to[out=80, in=180] (0.25, 1);
    \draw[color={rgb,255:red,0; green,100; blue,0}, thick, smooth] (0.75, 1) to[out=-10, in=100] (1, 0.75);
    \draw[color={rgb,255:red,0; green,100; blue,0}, thick, smooth] (0, 0.2) to[out=-80, in=180] (0.2, 0);
    \draw[color={rgb,255:red,0; green,100; blue,0}, thick, smooth] (0.8, 0) to[out=-5, in=-80] (1, 0.25);
    \draw[color={rgb,255:red,0; green,100; blue,0}, thick] (0.25, 1) -- (0.75, 1);
    \draw[color={rgb,255:red,0; green,100; blue,0}, thick] (0, 0.2) -- (0, 0.8);
    \draw[color={rgb,255:red,0; green,100; blue,0}, thick] (0.2, 0) -- (0.8, 0);
    \draw[color={rgb,255:red,0; green,100; blue,0}, thick] (1, 0.25) -- (1, 0.75);

    \filldraw[fill=gray, draw=black] (0, 0.8) circle (0.3pt) node[left,font=\tiny] {$\lambda_{\max} (\overline{\Pi}_{\VT} \Pi_{\SF} \overline{\Pi}_{\VT})$};
    \filldraw[fill=gray, draw=black] (0, 0.2) circle (0.3pt) node[left,xshift=0.25em,font=\tiny] {$1\!\!-\!\!\lambda_{\max} (\overline{\Pi}_{\VT} \overline{\Pi}_{\SF} \overline{\Pi}_{\VT})$};
    \filldraw[fill=gray, draw=black] (0.2, 0) circle (0.3pt) node[below left,xshift=1.7em,font=\tiny] {$1\!\!-\!\!\lambda_{\max} (\overline{\Pi}_{\SF} \overline{\Pi}_{\VT} \overline{\Pi}_{\SF})$};
    \filldraw[fill=gray, draw=black] (0.8, 0) circle (0.3pt) node[below left,xshift=2.05em,font=\tiny] {$\lambda_{\max} (\overline{\Pi}_{\SF} \Pi_{\VT} \overline{\Pi}_{\SF} )$};\
    \filldraw[fill=gray, draw=black] (0.25, 1) circle (0.3pt) node[above right,xshift=-2.8em,font=\tiny] {$1\!\!-\!\!\lambda_{\max} (\Pi_{\SF} \overline{\Pi}_{\VT} \Pi_{\SF})$};
    \filldraw[fill=gray, draw=black] (0.75, 1) circle (0.3pt) node[above right=-0.01em,xshift=0.15em, font=\tiny] {$\lambda_{\max} (\Pi_{\SF} \Pi_{\VT} \Pi_{\SF})$};
    \filldraw[fill=gray, draw=black] (1, 0.75) circle (0.3pt) node[right,xshift=-0.1em,font=\tiny] {$\lambda_{\max} (\Pi_{\VT} \Pi_{\SF} \Pi_{\VT} )$};
    \filldraw[fill=gray, draw=black] (1, 0.25) circle (0.3pt) node[right,xshift=-0.1em, font=\tiny] {$1\!\!-\!\!\lambda_{\max} (\Pi_{\VT} \overline{\Pi}_{\SF} \Pi_{\VT} )$};

    \draw[dashed] (0, 0.8) -- (1, 0.8);
    \draw[dashed] (0, 0.75) -- (1, 0.75);
    \draw[dashed] (0, 0.2) -- (1, 0.2);
    \draw[dashed] (0, 0.25) -- (1, 0.25);
    \draw[dashed] (0.2, 0) -- (0.2, 1);
    \draw[dashed] (0.8, 0) -- (0.8, 1);
    \node[below] at (1, 0) {\scriptsize $1$};
    \node[left] at (0, 1) {\scriptsize $1$};
\end{tikzpicture}
\caption{Feasible region $\Theta$ of $\alpha_{\VT}$ and $\beta_{\SF}$}
\label{fig:feasible region}
\end{figure}

\Cref{fig:feasible region} depicts the feasible region $\Theta$. The four boundary curves arise by turning the inequalities in \cref{eq.Theorem_feasible region} into equalities; these curves form the four corner arcs shown in the figure. In particular, the upper‑right boundary identifies pairs $(\alpha_{\VT},\beta_{\SF})$ that achieve the largest possible joint concentration in the vertex–time and spectral–frequency domains. 
This boundary is characterized by the equality
$\cos^{-1} \alpha_{\VT} + \cos^{-1} \beta_{\SF}
= \cos^{-1}\sqrt{\lambda_{\max}\parens*{\Pi_{\SF}\,\Pi_{\VT}\,\Pi_{\SF}}}$.

Typically, for a fixed spectral–frequency subspace $\calH_{\SF}$, enlarging the vertex–time subspace $\calH_{\VT}$ expands the feasible region and moves its upper‑right boundary toward the upper‑right corner. Here, enlarging a subspace is defined through strict inclusion, i.e., replacing $\calH_{\VT}$ with a larger subspace $\calH_{\VT}^{'}$ satisfying $\calH_{\VT}\subsetneq\calH_{\VT}^{'}$. The boundary curve collapses to that corner precisely when the projectors $\Pi_{\VT}$ and $\Pi_{\SF}$ admit perfect localization (see \cref{Sect.Localization properties}). 
Note that, under our subspace formulation, both $\calH_{\SF}$ and $\calH_{\VT}$ may have infinite dimension. More generally, any of the four boundary arcs in \cref{fig:feasible region} can collapse to its corresponding corner whenever the associated projector composition satisfies the perfect‑localization condition. Finally, it is instructive to compare the joint feasible region in \cref{eq.Theorem_feasible region} with the analogous regions that arise when only the vertex domain or only the time domain is considered.

\subsection{Special case of vertex-time and spectral-frequency subsets with product structure}

In the special case of \cref{eq.joint_VT_opt} and \cref{joint_fre_opt}, if the vertex-time subset $\calS$ has a product structure, then the vertex-time spread is bounded by the product of the vertex spread and time spread. A similar bound holds for the spectral-frequency spread.

\begin{Lemma}\label{Lemma.joint_bounds}
Given $\calS=\calV'\times \calT' \subseteq \calV\times \calT$ and $\Sigma = \Lambda'\times \Omega'\subseteq \Lambda\times \Omega$, let the vertex spread at a time $t$ and the time spread at a vertex $v$ be
\begin{align}
\label{eq.vertex_time_spread}
\alpha^{2}_{\calV'}(t) = \frac{\norm{f(\cdot, t)}_{L^{2}(\calV')}^{2}}{\norm{f(\cdot,t)}_{L^{2}(\calV)}^{2}}
~\text{and}~
\alpha^2_{\calT'}(v)= \frac{\norm{f(v,\cdot)}_{L^{2}(\calT')}^{2}}{\norm{f(v,\cdot)}_{L^{2}(\calT)}^{2}}
\end{align}
respectively.
Let the spectral spread at a (time) frequency $\omega$ and the frequency spread at a (graph) spectrum $\lambda$ be
\begin{align}
\label{eq.spec_fre_spread}
\beta^{2}_{\Lambda'}(\omega) = \frac{\norm{\calU f(\cdot,\omega)}_{L^{2}(\Lambda')}^{2}}{\norm{\calU f(\cdot,\omega)}_{L^{2}(\Lambda)}^{2}}
~\text{and}~
\beta^{2}_{\Omega'}(\lambda)=\frac{\norm{\calU f(\lambda,\cdot)}_{L^{2}(\Omega')}^{2}}{\norm{\calU f(\lambda,\cdot)}_{L^{2}(\Omega)}^{2}}
\end{align}
respectively.
The vertex-time spread $\alpha_{\calS}^2$ and spectral-frequency spread $\beta_{\Sigma}^2$ are bounded as follows:
\begin{align}
\min_{v}\alpha^{2}_{\calT'}(v)\inf_{t}\alpha^2_{\calV'}(t)\leq \alpha^2_{\calS} \leq \max_{v}\alpha^{2}_{\calT'}(v)\sup_t\alpha^2_{\calV'}(t) 
\label{eq.alpha_range}
\end{align}
and 
\begin{align}
\min_{\lambda} \beta^{2}_{\Omega'}(\lambda)\inf_{\omega}\beta_{\Lambda'}^{2}(\omega)\leq \beta^{2}_{\Sigma}\leq  \max_{\lambda} \beta^{2}_{\Omega'}(\lambda)\sup_{\omega}\beta_{\Lambda'}^{2}(\omega).
\label{eq.beta_sigma_range}
\end{align}
\end{Lemma}
\begin{IEEEproof}   
See \cref{sec.Lemma.joint_bounds}.
\end{IEEEproof}

We assume that the supremum and infimum in \cref{Lemma.joint_bounds} are achievable, and use the notations $\alpha_{\calT'}^{\text{min}}$, $\alpha^{\text{min}}_{\calV'}$, $\alpha_{\calT'}^{\text{max}}$, and $\alpha_{\calV'}^{\text{max}}$ as abbreviations for $\min_{v}\alpha_{\calT'}(v)$,  $\inf_t\alpha_{\calV'}(t)$, $\max_{v}\alpha_{\calT'}(v)$ and $\sup_t\alpha_{\calV'}(t)$, respectively. Similarly, we define $\beta_{\Lambda'}^{\text{min}}$, $\beta^{\text{min}}_{\Omega'}$, $\beta_{\Lambda'}^{\text{max}}$, and $\beta_{\Omega'}^{\text{max}}$ as $\inf_{\omega}\beta_{\Lambda'}(\omega)$, $\min_{\lambda} \beta_{\Omega'}(\lambda)$, $\sup_{\omega}\beta_{\Lambda'}(\omega)$, and $\max_{\lambda} \beta_{\Omega'}(\lambda)$, respectively.

\begin{Theorem}
\label{Theo.Feasible_Regoin_alpha}
Consider a special case $\Pi_{\VT}=\Pi_{\calS}$ and $\Pi_{\SF}=\Pi_{\Sigma}$ and such two subsets $\calS=\calV'\times \calT' \subseteq \calV\times \calT$ and $\Sigma = \Lambda'\times \Omega'\subseteq \Lambda\times \Omega$.  
Suppose $(\alpha_{\calS}, \beta_{\Sigma})\in\Theta$. Given the spectral spread $\beta^2_{\Lambda'}$ and the frequency spread $\beta^{2}_{\Omega'}$, defined in \cref{eq.spec_fre_spread}, the feasible region $\Theta_{\alpha_{\calS}}$ of $\alpha_{\calS}$ is given by
\begin{subequations}
\begin{align}
&\alpha_{\calS} \leq \cos\Big(\cos^{-1}\sqrt{\lambda_{\max}\left(\Pi_{\Lambda'}\Pi_{\calV'}\Pi_{\Lambda'}\right)\lambda_{\max}\left(\Pi_{\Omega'}\Pi_{\calT'}\Pi_{\Omega'}\right)}\notag\\
&\qquad\qquad\qquad - \cos^{-1}\left(\beta_{\Lambda'}^{\min}\beta_{\Omega'}^{\min}\right)\Big)\label{eq:ThetaS-1a}\\
&\alpha_{\calS} \leq \cos\Big(\cos^{-1}\sqrt{\lambda_{\max}\left(\overline{\Pi}_{\Lambda'}\Pi_{\calV'}\overline{\Pi}_{\Lambda'}\right)\lambda_{\max}\left(\overline{\Pi}_{\Omega'}\Pi_{\calT'}\overline{\Pi}_{\Omega'}\right)}\notag\\
&\qquad\qquad\qquad - \cos^{-1}\Big(\sqrt{1-(\beta_{\Lambda'}^{\max}\beta^{\max}_{\Omega'})^2}\Big)\Big),\\
& \alpha_{\calS} \geq \sin\Big(\cos^{-1}\sqrt{\lambda_{\max}\left(\Pi_{\Lambda'}\overline{\Pi}_{\calV'}\Pi_{\Lambda'}\right)\lambda_{\max}\left(\Pi_{\Omega'}\overline{\Pi}_{\calT'}\Pi_{\Omega'}\right)} \notag\\
&\qquad\qquad\qquad - \cos^{-1}\left(\beta_{\Lambda'}^{\min}\beta_{\Omega'}^{\min}\right)\Big),\\
& \alpha_{\calS} \geq \sin\Big(\cos^{-1}\sqrt{\lambda_{\max}\left(\overline{\Pi}_{\Lambda'}\overline{\Pi}_{\calV'}\overline{\Pi}_{\Lambda'}\right)\lambda_{\max}\left(\overline{\Pi}_{\Omega'}\overline{\Pi}_{\calT'}\overline{\Pi}_{\Omega'}\right)}\notag\\
&\qquad\qquad\qquad- \cos^{-1}\Big(\sqrt{1-\left(\beta_{\Lambda'}^{\max}\beta^{\max}_{\Omega'}\right)^2}\Big)\Big)
\end{align}
\end{subequations}
where $\lambda_{\max}(\cdot)$ denotes the largest eigenvalue of its operator argument.
\end{Theorem}
\begin{IEEEproof}
See \cref{sec:proof_Th2_alpha}.
\end{IEEEproof}

\cref{Theo.Feasible_Regoin_alpha} generalizes \cite{Tsitsvero2016} by extending vertex-domain localization to vertex–time localization, and it bounds the vertex–time spread $\alpha_{\calS}$ in terms of the spectral spread
$\beta_{\Lambda'}$ and the frequency spread $\beta_{\Omega'}$. 
Although similar spread bounds can be formulated in a discrete-time setting, the continuous-time formulation offers greater flexibility: it permits localization over \emph{arbitrary measurable time intervals}, whereas discrete time only allows localization over subsets of a fixed sampling grid common to all nodes. 
For cases of interest, where we assume that observations at different nodes are not synchronized, creating a common sampling grid across nodes would require selecting arbitrarily small sampling intervals or introducing some temporal distortion if the timing of some observations does not match sampling times. 

In contrast, the continuous-time setting allows localization over arbitrary measurable intervals, yielding resolution-independent projection operators and enabling the rigorous derivation of the spread bound in \cref{eq:ThetaS-1a}.
In particular, the upper bound in \cref{eq:ThetaS-1a} provides a theoretical criterion for selecting the subset $\calS$ with maximal energy concentration. The selected subset $\mathcal{S} $ is then used to design energy-concentrated bases, which are applied to discrete observations for signal reconstruction. Moreover, by combining the range of the vertex-time spread $\alpha^2_{\calS}$, as given in \cref{eq.alpha_range}, we have
\begin{align}
\begin{aligned}\label{eq.beta_range_special_case}
\small
&\beta_{\Sigma} \leq \cos\Big(\cos^{-1}\sqrt{\lambda_{\max}\left(\Pi_{\Lambda'}\Pi_{\calV'}\Pi_{\Lambda'}\right)\lambda_{\max}\left(\Pi_{\Omega'}\Pi_{\calT'}\Pi_{\Omega'}\right)}\\
&\qquad\qquad-\cos^{-1}\left(\alpha_{\calV'}^{\min}\alpha_{\calT'}^{\min}\right)\Big).
\end{aligned}   
\end{align}
While time-frequency localization is well understood, vertex-spectral localization can exhibit different behaviors depending on the graph structure, which may affect the tightness of the bound; nevertheless, it still provides a principled and informative criterion for graph topology inference, allowing us to promote energy concentration in the spectral-frequency domain.

\section{Vertex-time graph signal reconstruction}\label{sec.joint_dic}

Based on the vertex-time uncertainty principle, we propose a dictionary of signals that exhibit maximal energy concentration in the vertex-time domain. We characterize the localization properties of these signals and establish conditions for perfect localization within prescribed closed subspaces defined on the vertex-time and spectral-frequency domains. We then present methods for constructing and optimizing these atoms.

\subsection{Localization properties}\label{Sect.Localization properties}

A vertex-time signal $f\in L^{2}(\calV\times\calT)$ is said to be perfectly localized in the closed vertex-time subspace $\calH_{\VT}$ if
\begin{align}
\Pi_{\VT} f = f,
\label{eq.vertex_time_limited_signals}
\end{align}
where $\Pi_{\VT}$ is the orthogonal projection onto $\calH_{\VT}$ (cf.\ \cref{def.vertex-time_projection}).  Similarly, $f$ is perfectly localized in the spectral–frequency subspace $\calH_{\SF}$ if
\begin{align}
\Pi_{\SF} f = f,
\label{eq.spec_fre_limited_signals}
\end{align}
with $\Pi_{\SF}$ the spectral–frequency projector from \cref{def.joint_fre_opt_subspace}.  In practice, $\calH_{\SF}$ is often a support‑limited subspace, i.e., functions whose joint spectral–frequency content is confined to a measurable region $\Sigma\subseteq\Lambda\times\Omega$; signals satisfying \eqref{eq.spec_fre_limited_signals} are also referred to as joint bandlimited signals.

Under our definition, a vertex-time graph signal can achieve perfect localization simultaneously in the \emph{vertex-time} and \emph{spectral-frequency} domains with suitable choices of $\calH_{\VT}$ and $\calH_{\SF}$. This dual localization property is formalized in the following theorem.

\begin{Theorem}
\label{Theo:localized_both_domains}
A vertex-time graph signal $f$ is perfectly localized over both the vertex-time subspace $\calH_{\VT}$ and spectral-frequency subspace $\calH_{\SF}$ (i.e., $f\in \ima(\Pi_{\VT}) \cap \ima(\Pi_{\SF})$) if and only if the Hermitian operator $\Pi_{\SF}\Pi_{\VT}\Pi_{\SF}$ (or equivalently, $\Pi_{\VT}\Pi_{\SF}\Pi_{\VT}$) satisfies
\begin{align}\label{per_local_condition}
\Pi_{\SF}\Pi_{\VT}\Pi_{\SF} f = f.
\end{align}
An equivalent condition is 
\begin{align}\label{eq.conditions_bi_localized}
\norm{\Pi_{\VT}\Pi_{\SF}}=\norm{\Pi_{\SF}\Pi_{\VT}}=1,
\end{align}
where $\norm{}$ is the operator norm corresponding to $L^2(\calV\times \calT)$.
\end{Theorem}
\begin{IEEEproof}
See \cref{sec:proof_the3}.
\end{IEEEproof}

To illustrate \cref{Theo:localized_both_domains}, consider the following special cases:
\begin{enumerate}[i)]
\item In the case where $f\in L^2(\calV)$ (traditional GSP), a perfectly localized $f$ is an eigenvector of the operator $\Pi_{\Lambda'}\Pi_{\calV'}\Pi_{\calW'}$ associated with a unit eigenvalue, aligning with \cite[Theorem 2.1]{Tsitsvero2016}.

\item When the graph $\calG$ consists of a single vertex, the limiting operators $\Pi_{\SF}$ and $\Pi_{\VT}$ reduce to the frequency-limiting operator $\Pi_{\Omega'}$ and the time-limiting operator $\Pi_{\calT'}$. If $\calT'\subset\Real$ and $\Omega'= [-W,W] \subset \Real$ are both bounded sets, the condition in \cref{per_local_condition} no longer holds because the largest eigenvalue $\lambda_{\max}(\Pi_{\Omega'}\Pi_{\calT'}\Pi_{\Omega'})$ is strictly less than $1$ (see \cite{Pollak1961}). This result illustrates the classical time-frequency uncertainty principle: no nonzero signal can simultaneously have both a finite time duration and a finite frequency bandwidth.

\item Consider two subsets $\calS=\calV' \times \calT' \subseteq \calV \times \calT$ and $\Sigma = \Lambda' \times \Omega'\subseteq \Lambda \times \Omega$, where $\calT'\subset\Real$ and $\Omega'\subset \Real$ are both bounded sets. The vertex-time limiting operator $\Pi_{\calS}$ and the spectral-frequency limiting operator $\Pi_{\Sigma}$ decompose as $\Pi_{\calS}=\Pi_{\calV'}\otimes\Pi_{\calT'}$ and $\Pi_{\Sigma}=\Pi_{\Lambda'}\otimes\Pi_{\Omega'}$, respectively. Then $\Pi_{\Sigma}\Pi_{\calS}\Pi_{\Sigma}=\left(\Pi_{\Lambda'}\otimes\Pi_{\Omega'}\right)\left(\Pi_{\calV'}\otimes\Pi_{\calT'}\right)\left(\Pi_{\Lambda'}\otimes\Pi_{\Omega'}\right)$ and its largest eigenvalue factors as 
\begin{align*}
    \lambda_{\max}(\Pi_{\Sigma}\Pi_{\calS}\Pi_{\Sigma})=\lambda_{\max}(\Pi_{\Lambda'}\Pi_{\calV'}\Pi_{\Lambda'})\lambda_{\max}(\Pi_{\Omega'}\Pi_{\calT'}\Pi_{\Omega'}),
\end{align*}
which is less than $1$ as $\lambda_{\max}(\Pi_{\Omega'}\Pi_{\calT'}\Pi_{\Omega'})<1$. 
Consequently, a vertex-time graph signal $f$ cannot be perfectly localized over both the vertex-time subset $\calS$ and the spectral-frequency subset $\Sigma$ simultaneously.
\end{enumerate}


When perfect localization in both $ \calH_{\VT} $ and $ \calH_{\SF} $ is not achievable, we seek signals that are perfectly localized in one domain while maximizing energy concentration in the other. A common instance is to construct an orthonormal family of bandlimited functions (i.e., satisfying $ \Pi_{\Sigma} f = f$ in \cref{joint_fre_opt}) that are as concentrated as possible in the vertex–time subspace $ \calH_{\VT} $. This is obtained by solving the following sequence of constrained optimization problems for $ \{\xi_i\}_{i\ge 1} $:
\begin{align}\label{eq.optimization_problem}
\begin{aligned}
\xi_{i} = &\argmax_{\xi_{i}}\norm{\Pi_{\VT} \xi_{i}}\\
\ST & \norm{\xi_{i}} = 1,\  
\Pi_{\Sigma}\xi_{i}=\xi_{i},\\ 
&\angles{\xi_{i},\xi_{j}}=0,\; j = 1,\dots,i-1,\; i>1.  
\end{aligned}
\end{align}
Accordingly, $ \xi_1 $ is the bandlimited function with the largest $ \calH_{\VT} $ energy, $ \xi_2 $ is the next bandlimited function orthogonal to $ \xi_1 $ with maximal $ \calH_{\VT} $ energy, and so on. The next theorem characterizes the solution to \cref{eq.optimization_problem}.

\begin{Theorem}\label{Theo.solu_optimization}
A set of orthonormal joint bandlimited signals $\set{\xi_{i}}_{i\geq 1}$, each of which is maximally concentrated over a vertex-time subspace $\calH_{\VT}$, consists of the eigenvectors of the operator $\Pi_{\Sigma}\Pi_{\VT}\Pi_{\Sigma}$, i.e.,
\begin{align}
\Pi_{\Sigma}\Pi_{\VT}\Pi_{\Sigma} \xi_{i} = \lambda_{i} \xi_{i},
\end{align}
where the eigenvalues $\lambda_{1}\geq\lambda_{2}\geq\cdots$ form a non-increasing sequence whose only possible limit is zero. Additionally, for all $i, j\geq 1$,
\begin{align}
\angles{\xi_{i},\Pi_{\VT}\xi_{j}} = \lambda_{j}\delta_{ij},
\end{align}
where $\delta_{ij}$ is the Kronecker symbol.
\end{Theorem}
\begin{IEEEproof}
See \cref{sec:proof_theo_solution}.
\end{IEEEproof}

\cref{Theo.solu_optimization} provides a general framework for constructing a set of orthonormal bandlimited signals $\set{\xi_{i}}_{i=1,2,\dots}$ where $\xi_{i}\in L^{2}(\calV\times\calT)$. Specific cases include:
\begin{enumerate} [i)]
\item In traditional GSP, the set of orthonormal bandlimited signals corresponds to the set of orthonormal eigenvectors of $\Pi_{\Lambda'}\Pi_{\calV'}\Pi_{\Lambda'}$, as stated in \cite[Theorem 2.3]{Tsitsvero2016}.

\item  Suppose the graph $\calG$ consists of a single vertex and the limiting projection $\Pi_{\calS}$ reduces to the time-limiting operator $\Pi_{\calT'}$. The orthonormal joint bandlimited signals are then the prolate spheroidal wave functions $\set{\psi_{i}}_{i\geq1}$ \cite{Slepian1961,Laudau1962,Pollak1961}, achieved by the eigenfunctions of $\Pi_{\Omega'}\Pi_{\calT'}\Pi_{\Omega'}$ defined in the Hilbert space $L^{2}(\Real)$.

\item Consider the subsets $\calS=\calV'\times \calT' \subseteq \calV\times \calT$ and $\Sigma = \Lambda'\times \Omega'\subseteq \Lambda\times \Omega$. The vertex-time limiting operator $\Pi_{\calS}$ and the spectral-frequency limiting operator $\Pi_{\Sigma}$ can be expressed as $\Pi_{\calS}=\Pi_{\calV'}\otimes\Pi_{\calT'}$ and 
$\Pi_{\Sigma}=\Pi_{\Lambda'}\otimes\Pi_{\Omega'}$, respectively. The set of orthonormal joint bandlimited signals is given by $\set{\xi_{i,j}}_{i,j\geq1} = \set{\phi_{i}\otimes\psi_{j}}_{i,j\geq1}$. Here, $\set{\phi_{i}}_{i\geq1}$ are the orthonormal bandlimited vectors given by the eigenvectors of $\Pi_{\Lambda'}\Pi_{\calV'}\Pi_{\Lambda'}$ in the vertex domain, and $\set{\psi_{j}}_{j\geq1}$ are the orthonormal bandlimited signals (prolate spheroidal wave functions), determined by the eigenfunctions of $\Pi_{\Omega'}\Pi_{\calT'}\Pi_{\Omega'}$ defined in the Hilbert space $L^{2}(\Real)$.
\end{enumerate}


For numerical implementation, we parametrize the vertex–time subspace $\calH_{\VT}$. We study two representative parameterizations: (i) a support‑limited subspace (see \cref{subsec.supp_limite_subspace}), and (ii) a Gaussian heat‑kernel subspace (See \cref{subsec.Gaussian_subspace}). Both constructions yield joint bandlimited functions that are highly concentrated in the vertex–time domain, and their parameters are chosen to improve reconstruction accuracy. In the support‑limited case, we assemble a dictionary of joint bandlimited atoms supported on selected vertex–time blocks and employ a dictionary‑learning procedure to select the most informative atoms. In the Gaussian case, we grid-search the spatial and temporal diffusion scales and choose the pair that maximizes a predefined concentration score; then, we reconstruct by projecting onto the resulting basis.

\subsection{Joint parametric dictionary learning with support-limited subspace for vertex-time signal reconstruction}
\label{subsec.supp_limite_subspace}

An overcomplete dictionary can enhance signal reconstruction by enabling sparse, flexible, and robust signal representations. Motivated by sensor networks in which different sensors may be active or provide measurements over distinct, possibly non-overlapping time periods, empirical observations indicate that signal energy is not uniformly distributed across vertices or over time. Instead, it tends to concentrate in a limited region of the vertex–time domain (cf.\ \cref{fig:heat_map_covid_19}). To capture this localized structure while maintaining computational feasibility, we consider a structured low-dimensional parametric family of atoms with a vertex–time support of the form $\calS=\calV'\times \calT'$, where $\calT'=[t_{c}-\ell/2,t_{c}+\ell/2]$ denotes a finite time interval centered at $t_{c}$ with duration $\ell$, and $\calV'\subseteq \calV$ is a vertex subset designed to capture the dominant energy localization behavior over $\calT'$. Under this model, atoms are not learned as free vectors; rather, they are generated analytically from parameters $(\calV',t_{c},\ell)$, so learning reduces to selecting $\calV'$ and $(t_{c},\ell)$ that best match the observed energy localization. This parameterization explicitly encodes where energy is expected to concentrate in the vertex–time plane, yielding energy-concentrated atoms for efficient reconstruction. 


Our objective is to learn a dictionary bandlimited to $\Sigma=\Lambda'\times\Omega'$ and maximally concentrated on $\calS$.
Let $\Phi(\calV')=\left(\phi_{k}(\cdot;\calV'\right)_{k\geq1}$, where for each $k\geq 1$, $\phi_{k}(\cdot;\calV'):\calV\rightarrow \Real$ is an eigenvector of $\Pi_{\Lambda'}\Pi_{\calV'}\Pi_{\Lambda'}$, which is a graph signal localized over $\calV'$. Let $\Psi(\calT')=\left(\psi_{n}(\cdot;\calT'\right)_{n\geq1}$, where for each $n\geq1$, $\psi_{n}(\cdot;\calT'): \calT \rightarrow \Real$ is a PSWF localized over $\calT'$. Each $\psi_{n}(\cdot;\calT')$ is an eigenfunction of $\Pi_{\Omega'}\Pi_{\calT'}\Pi_{\Omega'}$. We construct a dictionary candidate as
\begin{align}
\label{eq.dictionary_atoms}
\Xi(\calV',\bt_{c}, \boldell) = \Phi(\calV')\otimes\Psi(\calT'),
\end{align}
where 
\begin{align*}
\Phi(\calV')\otimes\Psi(\calT')=\parens*{\phi_{k}(\cdot;{\calV'}) \otimes \psi_{n}(\cdot;{\calT'})}_{k,n\geq1}
\end{align*}
is the set of signals localized over $\calS$. 
While dictionaries with analytical forms and predefined parameters capture global signal structure for reconstruction, learned dictionaries with optimized parameters outperform by adapting to observed signals \cite{Elad2006,Mairal2009,Rencker2019}. 

Our goal is to optimize these parameters $\calV'$, $t_{c}$ and $\ell$ to construct an energy-concentrated vertex–time dictionary that enables accurate and efficient signal reconstruction using training samples. We adopt a train-test protocol: the dictionary parameters $(\calV',t_c,\ell)$ are learned from the training samples, and the resulting dictionary is then fixed and used to reconstruct held-out testing signals. We assume the training and testing data are generated by the same underlying process, so that they share similar vertex--time energy localization patterns and consistent spectral/frequency characteristics induced by the same graph topology and bandwidth-related priors. Let $\calM_{\mathrm{train}} = \set{(v_{m},t_{m})\given m=1,\dots,M}$ be the \textit{discrete} vertex-time instants at which the noisy observations $y_{m}=f(v_{m},t_{m})+\epsilon_{m}$, $m=1,\dots,M$ are available, where  
$\epsilon_{m}$ are \gls{iid} zero-mean random variables with variance $\sigma^{2}$. Let the training samples be denoted by $\bY_{\mathrm{train}}=\bigl(y_m\bigr)_{m=1,\dots,M}$.
We consider a general (possibly irregular) sampling model in which each $(v_m,t_m)$ corresponds to a single measurement event, and the index $m$ enumerates samples rather than time instants. In particular, $\calM_{\mathrm{train}}$ may be generated by a random sampling mechanism, so observations may occur at arbitrary vertex-time pairs and need not be synchronized across vertices. We assume the signal of interest is localized within a finite interval $[0,\delta]$. 
The dictionary $\Xi(\calV',t_{c}, \ell)$ consists of vertex-time dictionary atoms(columns), i.e.,\textit{ continuous time functions defined on $\calV\times[0,\delta]$, whose energy is concentrated on $\calV'\times\calT'$, where $\calT'=[t_c-\ell/2,\,t_c+\ell/2]$ and $\calV'\subseteq\calV$}. Evaluating these atoms at the training instances $\calM_{\mathrm{train}}$ yields a finite dictionary matrix, denoted by $\Xi(\calM_{\mathrm{train}};\calV', t_{c}, \ell)$, where each column corresponds to one of the atoms. The dictionary learning problem is as follows: 
\begin{problem}[Joint Energy Concentrated Dictionary Learning]\label{prob:joint_P}
 Given the training measurement set $\calM_{\mathrm{train}}$, and the noisy observations $\bY_{\mathrm{train}}$, the goal is to learn the energy-concentrated vertex--time dictionary (i.e., supported by \cref{eq:ThetaS-1a} in \Cref{Theo.Feasible_Regoin_alpha}) by jointly estimating the sparse coefficients $\bx$ and the dictionary parameters $(\calV', t_{c},\ell)$. This leads to the following optimization problem:
\begin{align}
\begin{aligned}
\label{eq.object_I=1}
\argmin_{\bx, \calV', t_{c}, \ell}\ & \norm{\bY_{\mathrm{train}}-\Xi(\calM_{\mathrm{train}}; \calV',t_{c}, \ell)\bx}_{2}^{2}+\mu\norm{\bx}_{1}\\
\ST\ &  t_{c} \in [0,\delta],\ 0 \leq \ell \leq \delta\\
&\calV'~\text{is optimized in}~\cref{eq:ThetaS-1a}~\text{with}~\abs{\calV'}\leq K\\
&\text{and}~ K=\rank(\Pi_{\Lambda'}). 
\end{aligned}
\end{align}
Here, both the spectral support $\Lambda'$ and the frequency support $\Omega'$ (implicitly included in \cref{eq:ThetaS-1a}) are treated as prior knowledge and are estimated from the full data observations $\bY$ via a cumulative-energy criterion. Specifically, we select the smallest set of graph and temporal frequencies whose cumulative energy accounts for at least a prescribed fraction $\beta_{\Lambda'}$ and $\beta_{\Omega'}\in (0,1]$ of the total signal energy respectively, i.e., $\Lambda' \in \argmin_{\Gamma \subseteq \Lambda} \abs{\Gamma}, \ST \frac{\sum_{\lambda \in \Gamma} \norm{\Pi_{\lambda} \bY}_2^2}{\norm{\bY}_2^2}\ge \beta_{\Lambda'},$ and 
$\Omega'\in \argmin_{\Gamma\subseteq\Omega} \text{meas}(\Gamma), \ST \frac{\int_{\Gamma}\|\Pi_{\omega}\bY\|_{2}^{2}\ud\omega}{\int_{\Omega}\|\Pi_{\omega}\bY\|_{2}^{2}\ud\omega}\ge \beta_{\Omega'}$, where $\text{meas}(\Gamma)$ denotes the Lebesgue measure of $\Gamma$.
The $\ell_1$ regularization term leads to sparse coefficients $\bx$, while $\mu$ controls the strength of the regularization.  
\end{problem}

Although \Cref{prob:joint_P} captures the ideal objective of jointly optimizing the vertex subset $\calV'$, the dictionary parameters $(t_{c},\ell)$, and the sparse coefficients $\bx$, it is highly nonconvex and combinatorial. In particular, the discrete optimization over $\calV'$ is tightly coupled with the continuous dictionary parameters $(t_{c},\ell)$, making the problem intractable to solve globally. To obtain a computationally feasible solution, we decompose \Cref{prob:joint_P} into two subproblems: (i) vertex subset selection (See \cref{vertex_subset_selection}), and (ii) dictionary learning via sparse coding with a fixed vertex subset $\calV'$, where the continuous dictionary parameters $(t_{c},\ell)$ are optimized (See \Cref{subsubsec.T.opt}). This decomposition decouples the discrete subset selection from continuous parameter optimization, yielding a principled yet heuristic approximation to the joint objective.

\subsubsection{Vertex subset selection}\label{vertex_subset_selection}

An effective dictionary design requires each atom to concentrate its energy on a carefully selected subset of vertices, while still preserving sufficient spectral coverage to capture the underlying graph structure. This balance between vertex-domain localization and spectral-domain expressiveness is essential for achieving accurate reconstruction using a compact and informative dictionary.
To this end, we select vertex subsets $\calV'$ to balance vertex localization and spectral coverage, ensuring that the spectral component continues to capture the underlying graph information while enabling accurate reconstruction with fewer, more informative atoms. Motivated by the upper bound of the feasible region in \cref{Theo.Feasible_Regoin_alpha}, which characterizes the $(\alpha_{\calS},\beta_{\Sigma})$ pairs achieving maximum localization in both vertex and spectral domains, the subset $\calV'$ selection is posed as an energy-concentration maximization problem, shown in \Cref{prob:vertex_selection}.

\begin{subproblem}[Vertex Subset Selection for Energy-Concentrated Dictionary Design]
Given the spectral spread $\beta_{\Lambda'}$, the frequency spread $\beta_{\Omega'}$,
a time interval $\mathcal T'$, and a bandwidth constraint $K=\rank(\Pi_{\Lambda'})$, the goal is to select a vertex subset $\calV'$ such that the resulting dictionary atoms are well localized in the vertex domain (i.e., maximizing energy concentration over vertex subset $\calV'$) while maintaining sufficient spectral coverage. This leads to the following vertex selection problem:
\label{prob:vertex_selection}
\begin{align*}
\begin{aligned}
\argmax_{\calV'}\ & \cos\Big(\cos^{-1}\sqrt{\lambda_{\max}\left(\Pi_{\Lambda'}\Pi_{\calV'}\Pi_{\Lambda'}\right)\lambda_{\max}\left(\Pi_{\Omega'}\Pi_{\calT'}\Pi_{\Omega'}\right)}\\&-\cos^{-1}\left(\beta_{\Lambda'}\beta_{\Omega'}\right)\Big)\\
\ST\ &  \abs{\calV'}\leq K, \ K=\rank(\Pi_{\Lambda'}).
\end{aligned}
\end{align*}  
\end{subproblem}

\begin{algorithm*}[!htb]
\caption{Greedy selection based on energy concentration maximization}\label{algo:max}
\begin{algorithmic}[1]
\State Input: Spectral spread $\beta_{\Lambda'}$, frequency spread $\beta_{\Omega'}$, time interval $\calT'$ and bandwidth $K=\rank(\Pi_{\Lambda'})$.
\State Output: Optimized $\calV'$.
\State Initialize $\calV' = \emptyset$.
\While{$\abs{\calV'} \leq K$}
\State $v^{\ast}=\argmax_{v_{i}} \cos\left(\cos^{-1}\sqrt{\lambda_{\max}\left(\Pi_{\Lambda'}\Pi_{\calV'\cup {\{v_{i}\}}}\Pi_{\Lambda'}\right)\lambda_{\max}\left(\Pi_{\Omega'}\Pi_{\calT'}\Pi_{\Omega'}\right)}-\cos^{-1}\left(\beta_{\Lambda'}\beta_{\Omega'}\right)\right)$
\State $\calV' \leftarrow \calV' \cup \{v^{\ast}\}$.
\EndWhile
\end{algorithmic}
\label{vertices selection}
\end{algorithm*}

\begin{algorithm*}[!htb]
\caption{Joint energy concentrated dictionary (JECD) learning algorithm}\label{algo:mini}
\begin{algorithmic}[1]
\State Input: Training instances $\calM_{\mathrm{train}}$, training samples $\bY_{\mathrm{train}}=\left(y_{m}\right)_{m\in\{1,\dots,M\}}$, spectral spread $\beta_{\Lambda'}$, frequency spread $\beta_{\Omega'}$, learning rate $\eta_{1}$ and $\eta_{2}$, convergence tolerance $\epsilon$, signal length $\delta$.
\State Output: Optimized $\widetilde{\calV}$, $\widetilde{t_{c}}$ and $\widetilde{\ell}$
\State Initialize $t_{c}^{0}$, $\ell^{0}$, $\calV' = \emptyset$, $u=0$.
\Repeat
\State $\calV'^{(u)}\leftarrow$ \cref{vertices selection} 
\State $\bx^{(u)} = \mathrm{Lasso}(\bY_{\mathrm{train}},\Xi(\calM_{\mathrm{train}}; \calV'^{(u)},t_{c}^{(u)}, \ell^{(u)}))$ 
\State $t_{c}^{(u+1)}=t_{c}^{(u)}-\eta_{1}\ppfrac{}{t_{c}^{u}}\norm{\bY_{\mathrm{train}}-\Xi(\calM_{\mathrm{train}};\calV'^{(u)}, t_{c}^{(u)}, \ell^{(u)})\bx^{(u)}}_{2}^{2}$
\State $\ell^{(u+1)}=\ell^{(u)}-\eta_{2}\ppfrac{}{\ell^{u}}\norm{\bY_{\mathrm{train}}-\Xi(\calM_{\mathrm{train}};\calV'^{(u)},t_{c}^{(u)}, \ell^{(u)})\bx^{(u)}}_{2}^{2}$
\If {$t_{c}^{(u+1)}\notin [0,\delta] $ \text{or} $\ell^{(u+1)}> \delta$}
\State $t_{c}^{(u+1)}=\max(0,\min(t_{c}^{(u+1)}, \delta))$
\State $\ell^{(u+1)}=\max(0,\min(\ell^{(u+1)}, \delta))$
\EndIf
\State $u\leftarrow u+1$
\Until{$\norm{\mathrm{Loss}^{(u)}-\mathrm{Loss}^{(u-1)}} \leq \epsilon$}
\end{algorithmic}
\end{algorithm*}

We use a greedy approach to solve \Cref{prob:vertex_selection}, as summarized in \cref{algo:max}. 
Given the selected vertex subset $\calV'$, we then optimize the continuous dictionary parameters $(t_{c},\ell)$ using the training samples to learn the dictionary $\Xi(\calV',t_{c},\ell)$ via sparse coding.


\subsubsection{Optimization of continuous dictionary parameters $(t_{c},\ell)$}\label{subsubsec.T.opt}
With $\calV'$ fixed, the dictionary learning problem reduces to a sparse coding problem in which only the temporal parameters $(t_c,\ell)$ and the sparse coefficients $\bx$ are optimized. This leads to \Cref{subprob.conti_learning}.

\begin{subproblem}\label{subprob.conti_learning}
Given the training measurement set $\calM_{\mathrm{train}}$, and the noisy observations $\bY_{\mathrm{train}}$, the goal is to learn the dictionary by jointly estimating the sparse coefficients $\bx$ and the dictionary parameters $(t_{c},\ell)$. This leads to the following optimization problem:
\begin{align}
\begin{aligned}\label{eq.caseI=N}
\argmin_{\bx, t_{c}, \ell}\ & \norm{\bY_{\mathrm{train}}-\Xi(\calM_{\mathrm{train}}; \calV',t_{c}, \ell)\bx}_{2}^{2}+\mu\norm{\bx}_{1}\\
\ST\ &  t_{c} \in [0,\delta],\ 0 \leq\ell\leq \delta.
\end{aligned}
\end{align}
Here, $\bx$ contains sparse coefficients, and $\mu$ is a regularization parameter. 
\end{subproblem}

The optimization problem \cref{eq.caseI=N} is not jointly convex but is convex with respect to each variable when the others are fixed. We optimize by alternating minimization over $\bx$, $t_{c}$ and $\ell$ until a convergence criterion. At the $u$-th iteration, the following minimizations are performed sequentially:
\begin{align*}
\bx^{u+1}&\leftarrow\argmin_{\bx^{u}} \norm{\bY_{\mathrm{train}}-\Xi(\calM_{\mathrm{train}};\calV',t_{c}^{u},\ell^{u})\bx^{u}}_{2}^{2}+\mu\norm{\bx^{u}}_{1} \\
t_{c}^{u+1}&\leftarrow\argmin_{t_{c}^{u}}\norm{\bY_{\mathrm{train}}-\Xi(\calM_{\mathrm{train}};\calV’,t_{c}^{u},\ell^{u})\bx^{u+1}}_{2}^{2} \\
&\qquad\qquad\ST t_{c} \in [0,\delta]\\
\ell^{u+1} &\leftarrow \argmin_{\ell^{u}} \norm{\bY_{\mathrm{train}}-\Xi(\calM_{\mathrm{train}};\calV',t_{c}^{u+1},\ell^{u})\bx^{u+1}}_{2}^{2}\\
&\qquad\qquad\ST 0 \leq \ell\leq \delta.
\end{align*}
The sparse coefficients $\bx$ are updated using Lasso regression\cite{Tibshirani1996}. After updating $\bx$, the dictionary is updated by optimizing $t_{c}$ and $\ell$ using the gradient descent algorithm \cite{Lecun1998}. We summarize the complete procedure in \cref{algo:mini}, referred to as joint energy concentrated dictionary (JECD), where Loss denotes the objective function in \cref{eq.object_I=1}.

Suppose the optimization solution of \Cref{prob:joint_P} is given by $(\widetilde{\calV}, \widetilde{t}_{c}, \widetilde{\ell})$. 
With the optimized dictionary $\Xi(\cdot;\widetilde{\calV},\widetilde{t}_{c}, \widetilde{\ell})$, the reconstruction process is formulated as:
\begin{align}\label{eq.reconstruction}
\argmin_{\bx} &\norm{\bY_{\mathrm{test}}-\Xi(\calM_{\mathrm{test}}; \widetilde{\calV}, \widetilde{t}_{c}, \widetilde{\ell})\bx}_{2}^{2}+\mu\norm{\bx}_{1}
\end{align}
where $\calM_{\mathrm{test}}$ denotes test instances and $\bY_{\mathrm{test}}$ the corresponding test sample values. Suppose the optimal solution in \cref{eq.reconstruction} is given by $\bx^*$. The reconstructed vertex-time graph signal is given by $\widehat{f}(v,t)=\Xi(\{(v,t)\}; \widetilde{\calV}, \widetilde{t}_{c}, \widetilde{\ell})\bx^*$ for each $(v,t)\in\calV\times\calT$. The recovery performance is measured by the relative square error:
\begin{align}
\label{recovery_error}
\mathrm{RSE} = \frac{\sum_{(v,t)\in \calM_\mathrm{val}}\left(f(v,t)-\widehat{f}(v,t)\right)^2}{\sum_{(v,t)\in \calM_\mathrm{val}} f(v,t)^2},
\end{align}
where $\calM_\mathrm{val}$ denotes a validation set of sample points.


\subsection{Gaussian subspace basis for vertex-time graph signal reconstruction}\label{sec.gaussian_subspace}
\label{subsec.Gaussian_subspace}
Consider the closed Gaussian heat–kernel subspace $\calH_{\VT} = \overline{\spn}\set{ e^{-\tau_{v}\bL}\delta_{v_{0}}\otimes \tfrac{1}{\sqrt{4\pi\tau_{t}}}e^{-\frac{(t-t_{0})^{2}}{4\tau_{t}}} \allowbreak\given\allowbreak v_{0}\in\calV',\; t_{0}\in\calT' }$, parameterized by diffusion scales $\tau_{v},\tau_{t}>0$. Here, $\delta_{v_{0}}$ is the Kronecker delta at vertex $v_{0}$, $e^{-\tau_{v}\bL}\delta_{v_{0}}$ is the graph heat atom spatially localized about $v_{0}$, and $\tfrac{1}{\sqrt{4\pi\tau_{t}}}e^{-(t-t_{0})^{2}/(4\tau_{t})}$ is the temporal Gaussian kernel centered at time $t_{0}$. Each atom is separable and centered at $(v_{0},t_{0})$, with $\tau_{v}$ and $\tau_{t}$ controlling the spatial and temporal spreads respectively. Let $\Pi_{\tau_{v},\tau_{t}}$ denote the orthogonal projector onto $\calH_{\VT}$. We seek an orthonormal family of jointly bandlimited signals that maximizes vertex–time concentration. Formally, $\xi_{i} = \argmax_{\xi_{i}}\ \norm{\Pi_{\tau_{v},\tau_{t}}\xi_{i}}$, $\ST \norm{\xi_{i}} = 1, \Pi_{\Sigma}\xi_{i}=\xi_{i}, \langle\xi_{i},\xi_{j}\rangle=0$ for $j<i$, where $\Pi_{\Sigma}$ enforces the joint spectral–frequency bandlimit.

In practice, we search over discrete candidate scales. Let $\calP_{v}=\set*{\tau_{v}^{(1)},\dots,\tau_{v}^{(P)}}$ and $\calP_{t}=\set*{\tau_{t}^{(1)},\dots,\tau_{t}^{(Q)}}$. For each pair $(\tau_{v}^{(p)},\tau_{t}^{(q)})$ form $\bA^{(p,q)} := \Pi_{\Sigma}\,\Pi_{\tau_{v}^{(p)},\tau_{t}^{(q)}}\,\Pi_{\Sigma}$, 
and solve the eigenproblem $\bA^{(p,q)}\xi=\lambda\xi$ on the subspace $\ima(\Pi_{\Sigma})$. Denote the nonincreasing eigenvalues by $\lambda_{1}^{(p,q)}\ge\lambda_{2}^{(p,q)}\ge\cdots\ge\lambda_{K}^{(p,q)}$, where $K=\operatorname{rank}(\bA^{(p,q)})$. Choose the optimal scales by maximizing a concentration score $\Phi$, for example, the total energy of the top $K$ eigenvalues:
\begin{align}\label{opt_diffusion_scales}   
(\tau_{v}^{\star},\tau_{t}^{\star}) \in \argmax_{(\tau_{v}^{(p)},\tau_{t}^{(q)})\in\calP_{v}\times\calP_{t}} \Phi(\lambda_{1}^{(p,q)},\dots,\lambda_{K}^{(p,q)}),
\end{align}
with $\Phi(\lambda_{1},\dots,\lambda_{K})=\sum_{k=1}^{K}\lambda_{k}$ as a canonical choice.

With the optimal pair $(\tau_{v}^{\star},\tau_{t}^{\star})$ fixed, let $\bA^{\star}=\bA^{(p^{\star},q^{\star})}$. The leading eigenvectors of $\bA^{\star}$ yield an orthonormal set $\{\xi_{k}^{\star}\}_{k=1}^{K}$ of jointly bandlimited atoms maximally concentrated in the chosen Gaussian subspace. Given a noisy observation $f_{\mathrm{obs}}(v,t)$, we reconstruct by projection: $\widehat{f}(v,t)=\sum_{k=1}^{K}\langle f_{\mathrm{obs}},\xi_{k}^{\star}\rangle\,\xi_{k}^{\star}(v,t)$ and evaluate performance using the RSE in \cref{recovery_error}.


\section{Graph topology inference using vertex-time samples}
\label{sec.gl}

In this section, we propose a graph inference approach leveraging the energy concentration concept stemming from the uncertainty principle. We 
model the signals as being approximately bandlimited in the spectral-frequency domain and concentrated in the vertex-time domain. Our aim is to learn a graph that maximizes spectral-frequency energy concentration, thereby capturing the essential structure of the signals while honoring the uncertainty principle’s trade-off between vertex and spectral localization.

Given a set of $M$ graph signals $\{\by_{m}\}_{m=1}^{M}$ collected at different time instances, we model each signal as approximately bandlimited, such that $\by_m \approx \bU \bs_m$, with negligible approximation error. Here, $\bU$ is an orthonormal matrix, viewed as the graph Fourier basis, and each $\bs_{m}$ is a sparse coefficient vector. We further design $\bU$ to maximize energy concentration in the joint spectral–frequency domain by maximizing the upper bound in \cref{eq.beta_range_special_case}, so the learned graph topology aligns with the observed signals while enforcing the desired energy concentration. Define the GFT coefficient matrix as $\bS \triangleq [\bs_{1},\dots, \bs_{M}]\in \Real^{N\times M}$. We assume that all coefficients share a common zero support, implying that $\bS$ is block-sparse with multiple zero rows. 

Specifically, let
$\calB_{K}\triangleq \bigl\{ \bS = [\bs_{1},\dots, \bs_{M}]\in \Real^{N\times M} : \bS(i,:)=\mathbf{0}\,\text{for all }i\notin \calK \subseteq \calV,\; K = |\calK| \bigr\}$,
where $\bS(i,:)$ is the $i$-th row of $\bS$, and each column $\bs_{m}$ represents the GFT of $\by_{m}$. Letting $\bY \triangleq [\by_{1}, \dots, \by_{M}]$ be the observation matrix, the model is compactly written as $\bY=\bU\bS$.

\begin{algorithm*}[!htb]
\caption{Energy concentrated graph learning (ECGL) algorithm}\label{algo:gl}
\begin{algorithmic}[1]
\State Input: Observations $\bY = [\by_{1}, \dots, \by_{M}]$, vertex spread $\alpha_{\calV'}$, time spread $\alpha_{\calT'}$, learning rate $\eta$, tolerance $\epsilon$, and weight $\gamma$.
\State Output: Optimized $\widetilde{\bL}$.
\State Initialize $\bU^{0}$, $u=0$.
\Repeat
\State $\bS^{u+1} = \Pi_{\calB_{K}}\bigl((\bU^{u})\T\bY\bigr)$ 
\State $\bU^{u+1} = \bU^{u}-\eta \Bigl(\ppfrac{}{\,\bU^{u}}\norm{(\bU^{u})\T\bY-\bS^{(u)}}_{F}^{2} + \gamma \ppfrac{}{\,\bU^{u}} \Bigl(\norm{\Pi_{\Lambda'}\Pi_{\calV'}\Pi_{\Lambda'}}_{2} - \frac{(\alpha_{\calV'}\alpha_{\calT'})^2}{\lambda_{\max}(\Pi_{\Omega'}\Pi_{\calT'}\Pi_{\Omega'})}\Bigr)^2 \Bigr)$
\State $\bL^{(u+1)} = \bU^{(u)}\,\Lambda\,(\bU^{(u)})\T$
\State $u \leftarrow u+1$
\Until{$\bigl|\!\norm{(\bU^{(u+1)})\T\bY-\bS^{(u+1)}}_{F}^{2} - \norm{(\bU^{(u)})\T\bY-\bS^{(u)}}_{F}^{2}\bigr| \leq \epsilon$}
\end{algorithmic}
\end{algorithm*}

Our objective is to jointly learn the orthonormal transform $\bU$, the sparse matrix $\bS$ representing the GFT coefficients of the signals, and the underlying graph topology, characterized by the Laplacian matrix $\bL$, which admits the columns of $\bU$ as its eigenvectors. This is formulated as the following optimization problem:
\begin{align}
\begin{aligned}\label{opt.gl}
&\min_{\bL, \bU, \Lambda \in \Real^{N\times N}, \bS\in \Real^{N\times M}} \norm{\bY-\bU\bS}_{F}^{2} + \gamma \Big(1-\cos\\
&\big(\cos^{-1}\sqrt{\lambda_{\max}\left(\Pi_{\Lambda'}\Pi_{\calV'}\Pi_{\Lambda'}\right)\lambda_{\max}\big(\Pi_{\Omega'}\Pi_{\calT'}\Pi_{\Omega'}\big)}\\
&-\cos^{-1}\big(\alpha_{\calV'}\alpha_{\calT'}\big)\big)\Big)+ \text{TV}(\bY,\bL) +  \mu \norm{\bL}_{F}^{2},\\
\ST\ 
& \Pi_{\calW'}=\bU\Sigma_{\calW'}\bU\T,\ 
\bU\bU\T=\bI, \bu_{1} =  \bone /\sqrt{N},\\
&\bL\bU = \bU\Lambda, \bL\in \calL,\ \tr(\bL)=p,\\
& \Lambda \succeq \textbf{0}, \Lambda_{ij}=\Lambda_{ji}=0, \forall i\neq j,\ \bS\in \calB_{K},
\end{aligned}
\end{align}
where $\calL$ denotes the class of graph Laplacian matrices, i.e., $\calL\triangleq \set{\bL = (L_{ij}) \in \calS_{+}^{N} \given \bL\textbf{1}=\textbf{0}, L_{ij}=L_{ji} \leq 0, \forall i\neq j}$, and $\calS_{+}^{N}$ is the set of real, symmetric and positive semidefinite matrices. The term $\text{TV}(\bL,\bY)$ represents the $\ell_{1}$-norm total variation of the observed signal $\bY$ on the graph, defined as $\text{TV}(\bL,\bY)\triangleq \sum_{m=1}^{M}\sum_{i,j=1}^{N} L_{ij} \abs{Y_{im}-Y_{jm}}$. 

The constraints in \cref{opt.gl} force $\bU$ to be unitary and require that it includes a vector proportional to the all-ones vector. The Laplacian matrix $\bL$ is constrained to be a valid combinatorial Laplacian, with its eigenvectors given by the columns of $\bU$ and its eigenvalues represented by the diagonal entries of $\Lambda$. The trace constraint $\text{tr}(\bL)=p>0$ is imposed to avoid the trivial solution. 

The optimization problem \cref{opt.gl} is not jointly convex but is convex with respect to each variable when the others are fixed. We solve the problem by alternating minimization over $\bS$, $\bU$ and $\bL$ until convergence. At the $u$-th iteration, the following minimizations are performed sequentially:
\begin{align*}
\bS^{(u+1)}&\leftarrow\argmin_{\bS^{(u)}\in \calB_{K}}~ \norm{\left(\bU^{(u)}\right)\T\bY-\bS^{(u)}}_{F}^{2}\\
\bU^{(u+1)}&\leftarrow\argmin_{\bU^{(u)}}~ \norm{\left(\bU^{(u)}\right)\T\bY-\bS^{(u)}}_{F}^{2}\\
&\qquad + \gamma \left(\norm{\Pi_{\Lambda'}\Pi_{\calV'}\Pi_{\Lambda'}}_{2}-\frac{\left(\alpha_{\calV'}\alpha_{\calT'}\right)^2}{\lambda_{\max}(\Pi_{\Omega'}\Pi_{\calT'}\Pi_{\Omega'})}\right)^2\\
&\qquad\ST \Pi_{\Lambda'}= \bU^{(u)} \Sigma_{\Lambda'}\left(\bU^{(u)}\right)\T,\
 \bU^{(u)}\left(\bU^{(u)}\right)\T=\bI\\
&\qquad\qquad~ \bu_{1}^{(u)} = \bone / \sqrt{N},\\
\bL^{(u+1)} &\leftarrow \argmin_{\bL^{(u)}} \text{tr}(\bY\bL^{(u)}\bY\T)+\mu \norm{\bL^{(u)}}_{F}^{2},\\
&\qquad\ST \bL^{(u)}\in\calL, \text{tr}(\bL^{(u)})=p,\
\bL^{(u)}\bU^{(u)}=\bU^{u}\Lambda,\\
&\qquad\qquad \Lambda \succeq \textbf{0}, \Lambda_{ij}=\Lambda_{ji}=0, \forall i\neq j.
\end{align*}

The complete procedure is summarized in \cref{algo:gl}, referred to as \emph{energy-concentrated graph learning (ECGL)}. With the optimization solution $\widetilde{\bL}$, we measure performance using the correlation coefficient between the true Laplacian entries $L_{ij}$ and their estimates $\tilde{L}_{ij}$:
\begin{align}
\label{gl_rho}
\rho(\bL,\widetilde{\bL}) \triangleq \frac{\sum_{ij}L_{ij}\tilde{L}_{ij}}{\sqrt{\sum_{ij}L_{ij}^2}\sqrt{\sum_{ij}\tilde{L}_{ij}^2}}.
\end{align}
We evaluate the topology inference using the average recovery error of the adjacency matrix:
\begin{align}
\label{gl_adj}
\bar{\varepsilon}_{F} \triangleq \frac{\norm{\bA-\widetilde{\bA}}_{F}}{N(N-1)},
\end{align}
where $\bA$ is the ground-truth binary adjacency matrix and $\widetilde{\bA}$ is derived from $\widetilde{\bL}$. We also adopt edge precision and recall \cite{Manning2008}, which quantify the proportion of correctly detected edges and the proportion of true edges recovered. Defining $\calE_{g}$ and $\calE_{r}$ as the ground-truth and recovered edge sets, respectively, these metrics are:
\begin{align}
\label{gl_pre_recall}
\text{Precision} = \frac{\abs{\calE_{g}\cap \calE_{r}}}{\abs{\calE_{r}}}, \quad
\text{Recall}= \frac{\abs{\calE_{g}\cap \calE_{r}}}{\abs{\calE_{g}}}.
\end{align}

\section{Experimental Results}\label{sec.exp}

\subsection{Vertex-time graph signal reconstruction}\label{Exper_setup}
We validate our theoretical developments via numerical experiments on three tasks. First, we evaluate JECD for vertex–time signal reconstruction and compare it against several baselines using two real datasets: the California COVID‑19 county case counts and the California traffic flow measurements. Second, we assess the Gaussian heat‑kernel subspace basis on synthetic vertex–time signals to quantify its reconstruction accuracy and noise robustness. Third, we test ECGL on a synthetic graph for topology inference and compare it with state‑of‑the‑art graph learning algorithms. 

\subsubsection{Support-limited subspace case}
We split each data set into training and test sets. The JECD $\Xi$ is trained using a proportion $p_{o}$ of the training samples. We then reconstruct the test set by using $p_{o}$ of its samples and measure performance on the remaining $1-p_{o}$. We compare our approach with the following vertex-time dictionaries for vertex-time graph signal reconstruction:
\begin{enumerate}[i)]
\item \Gls{JFT}\cite{Loukas2016}: The JFT dictionary is constructed from a set of joint vertex-time Fourier bases $\set{\phi_{k}\otimes\psi_{l} \given k=1,\dots,K,\ l=1,\dots,L}$. Here, $\phi_{k}$ is the $k$-th orthonormal eigenvector of the graph Laplacian $\bL$, and $\psi_{l}$ is an element of the orthonormal Fourier basis in the time domain, given by $\set{\frac{1}{\sqrt{2\pi}}, \frac{1}{\sqrt{\pi}}\cos(lt), \frac{1}{\sqrt{\pi}}\sin(lt)}_{l\geq1}$. Both $K$ and $L$ are tunable hyperparameters. 
\item \Gls{STVFT} \cite{Grassi2018}: We design the mother vertex-time kernel in the short time-vertex Fourier basis in a separable manner. In the vertex domain, following \cite{Grassi2018}, the localized graph basis is $\bh_{p,q} = (\bU h(\Lambda - \Lambda_{q}) \bU^{\ast}) \bdelta_{p}$, where $h(\cdot)$ is adopted as an Itersine kernel \cite{Perraudin2014}, given by $h(\Lambda) = \sin(0.5\pi\cos((\pi\Lambda)^2)$ with uniformly $Q$ translations $\Lambda_{q} = \diag(\frac{\lambda_{\max}}{Q}q)$ for $q = 1, \ldots, Q$. Here, $\bU$ is the unitary matrix whose columns are the Laplacian eigenvectors, $\Lambda$ is a diagonal matrix with non-negative real eigenvalues $\set{\lambda_{i}}_{i=1,\dots,N}$, and $\bdelta_p$ is a Kronecker delta centered at vertex $p$.
In the time domain, we use Gabor functions $g_{m,n}(t) = \frac{1}{\sqrt{2\pi}\rho}e^{-\frac{(t-m\tau_{0})^2}{2\rho^2}} e^{j n\omega_{0}}$ where $m, n \in \mathbb{Z}$, $\tau_{0}$ and $\omega_{0}$ are time and frequency shifts, and $\rho$ controls the Gaussian window width. 
The STVFT dictionary is constructed from the short-time-vertex Fourier bases $\set{\bh_{p,q}\otimes g_{m,n}(t)}_{p,q,m,n}$. The parameters $Q$, $\tau_{0}$, $\omega_{0}$, and $\sigma$ are tunable hyperparameters.
\item \Gls{STVWT} \cite{Grassi2018}: The mother vertex-time kernel in the spectral time-vertex wavelet basis is also designed in a separable manner. For the vertex dimension, the localized graph wavelet basis is $\bh_{p,s}=(\bU h(s\Lambda)\bU^{\ast})\bdelta_{p}$, where $h(s\Lambda)$ is the scaled Itersine kernel \cite{Perraudin2014} and $\bdelta_p$ is a Kronecker delta centered at vertex $p \in \calV$. For the time dimension, we use a Morlet wavelet $\psi_{a,b}(t) = \frac{1}{\sqrt{a}} e^{-\frac{(t-b)^2}{2a^2}} e^{j\omega_{0} \frac{t-b}{a}}$ with $a>0$ and $b\in \Real$ as the scale and translation parameters, controlling dilation and time shift respectively, and $\omega_{0}$ as the central frequency. The STVWT dictionary is constructed from a collection of $\set{\bh_{p,s}\otimes\psi_{a,b}(t)}_{p,s,a,b}$. Here, $s$, $a$ and $b$ are tunable hyperparameters.
\item Naive extension of GUP (NEGUP) \cite{Tsitsvero2016}: In the vertex domain, the graph localized basis is $\Phi(\calV')=\left(\phi_{k}(\cdot;\calV')\right)_{k\geq1}$, where for each $k\geq 1$, $\phi_{k}(\cdot;\calV'):\calV\rightarrow \Real$ is an eigenvector of $\Pi_{\Lambda'}\Pi_{\calV'}\Pi_{\Lambda'}$, which is a graph signal localized over $\calV'$. Here $\calV'$ is determined by the right upper bound of the feasible region, shown in Theorem 2.1 in \cite{Tsitsvero2016}, as follows
    \begin{align*}
        \argmax_{\calV'}\ &\cos\left(\cos^{-1}\sqrt{\lambda_{\max}\left(\Pi_{\Lambda'}\Pi_{\calV'}\Pi_{\Lambda'}\right)}-\cos^{-1}\beta_{\Lambda'}\right)\\
        \ST\ &  \abs{\calV'}\leq K, \ K=\rank(\Pi_{\Lambda'}).
    \end{align*}
In the time domain, the localized basis is $\Psi(\calT')=\left(\psi_{n}(\cdot;\calT')\right)_{n\geq1}$, where for each $n\geq1$, $\psi_{n}(\cdot;\calT'): \calT \rightarrow \Real$ is a PSWF localized over $\calT'$. Each $\psi_{n}(\cdot;\calT')$ is an eigenfunction of $\Pi_{\Omega'}\Pi_{\calT'}\Pi_{\Omega'}$. Here $\calT'$ is predefined by the signal's length, denoted as $\calT'=[-\ell/2,\ell/2]$. The NEGUP dictionary is constructed by $\Phi(\calV')\otimes\Psi(\calT')$.
\end{enumerate}

We evaluate vertex-time signal reconstruction on a COVID-19 dataset from The New York Times, which compiles reports from state and local health agencies.\footnote{\url{https://github.com/TorchSpatiotemporal/tsl}} We extract records for California’s 58 counties from the first day each reported cases, spanning July 29, 2020, to August 1, 2022. The data is split into a training set $\calM_{\mathrm{train}}$ (July 29, 2020–July 30, 2021) and a test set $\calM_{\mathrm{test}}$ (July 31, 2021–August 1, 2022). Each county is treated as a vertex, with edges connecting adjacent counties. We 
randomly select a proportion $p_o$ of the training samples to optimize the dictionary $\Xi$. We then reconstruct the test set using $p_o$ of its samples and evaluate the performance on the remaining $1-p_o$ proportion according to \cref{recovery_error}. 


\begin{figure}[!htbp]
\centering
\includegraphics[width=1\columnwidth]{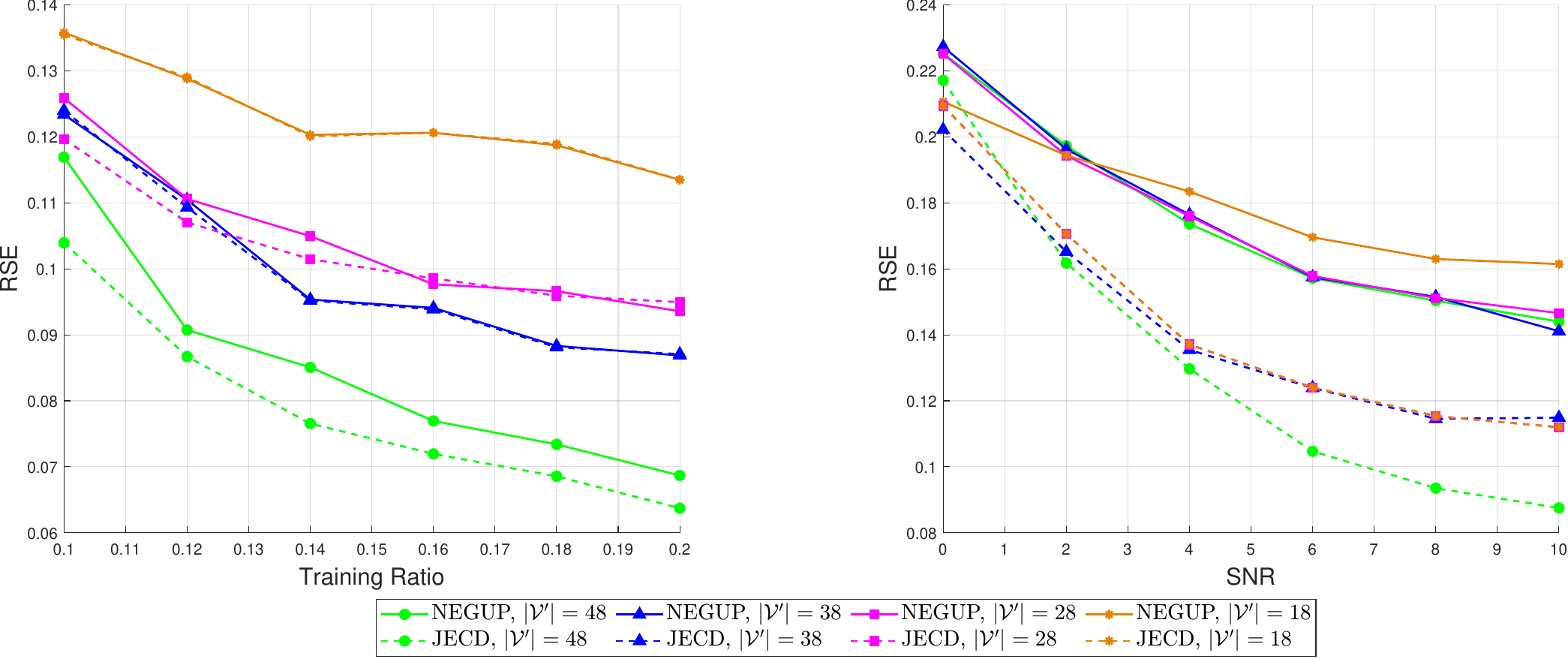}
\begin{subfigure}[b]{0.485\columnwidth}
    \phantomsubcaption
    \label{fig:Covid_NE_JECD_train_ratio}
    \vspace{1mm}
    \centering \footnotesize (a) Varying training ratios and no additive noise.
\end{subfigure}
\hfill
\begin{subfigure}[b]{0.485\columnwidth}
    \phantomsubcaption
    \label{fig:Covid_NE_JECD_SNR}
    \centering \footnotesize (b) Varying SNR and fixed training ratio of $0.2$.
\end{subfigure}
\caption{Reconstruction error over test set under different conditions. Each point in the figure is obtained by 10 repetitions.} 
\end{figure}


By setting different values of $\beta_{\Lambda'}$ at $99\%$, $97\%$, $95\%$, and $92.5\%$, the resulting sets $\calV'$ have cardinalities of $48$, $38$, $28$, and $18$, respectively, for both JECD and NEGUP. \Cref{fig:Covid_NE_JECD_train_ratio} and \Cref{fig:Covid_NE_JECD_SNR} show the effect of $\beta_{\Lambda'}$ on reconstruction performance across training ratios from $0.1$ to $0.2$ and SNR levels from $0$ dB to $10$ dB. JECD consistently outperforms NEGUP under all conditions, with the highest accuracy achieved when $\abs{\calV'}=48$. These findings confirm JECD’s ability to exploit temporal dependencies and adapt to time-varying data, resulting in a more effective dictionary design than NEGUP’s static graph-based approach.

To further highlight JECD's advantages, it is also compared against JFT, STVFT, and STVWT. \Cref{fig:Covid_NE_JECD_MSE} shows that JECD maintains the lowest RSE across various training ratios and SNR levels, demonstrating robust performance in noisy and dynamic scenarios and confirming its overall superiority among the tested methods.

\begin{figure}[!htbp]
\centering
\vspace{-4mm}
\begin{subfigure}[t]{0.49\columnwidth}
\centering
\includegraphics[width=\columnwidth, trim={3.5cm 9cm 3.5cm 9cm}, clip]{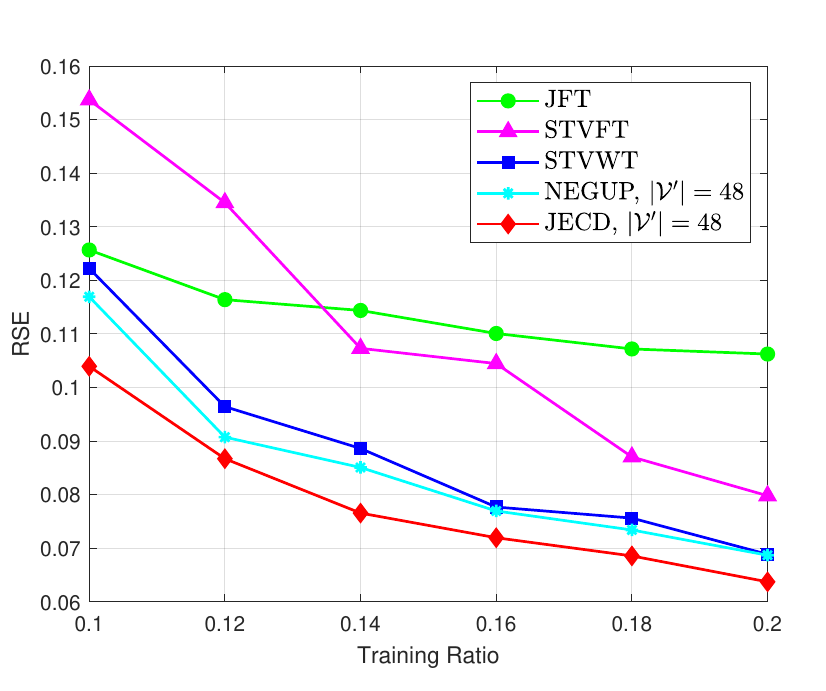}
\caption{RSE with varying training ratios and no additive noise. \vspace{-1em}}
\label{fig:Covid_MSE_TrainingRatio}
\end{subfigure}
\hfill
\begin{subfigure}[t]{0.49\columnwidth}
\centering
\includegraphics[width=\columnwidth, trim={3.5cm 9cm 3.5cm 9cm}, clip]{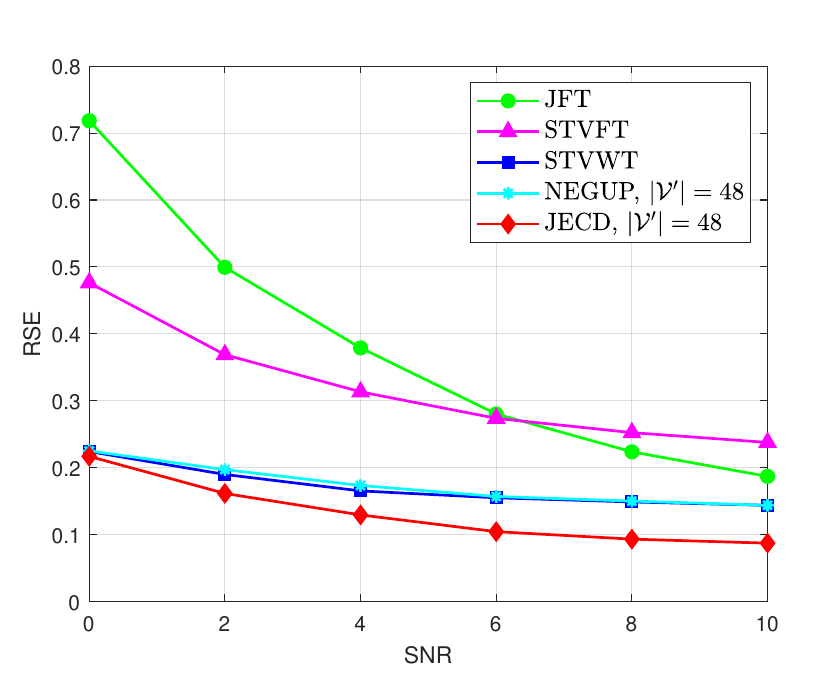}
\caption{RSE with varying SNR levels and a fixed training ratio of $0.2$.}
\label{fig:Covid_MSE_SNR}
\end{subfigure}
\caption{Reconstruction error under varying conditions. Each data point is the average of 10 experiments.}
\label{fig:Covid_NE_JECD_MSE}
\end{figure}

\subsubsection{Gaussian subspace case}\label{sec.exp_gaussian_basis}
We evaluate the Gaussian heat-kernel subspace basis on synthetic data. An Erd\H{o}s-R\`{e}nyi graph with $N=40$ nodes and connection probability $0.5$ is generated. The ground-truth vertex–time signal is $(K,\omega_{0})$-bandlimited and defined by $y = \sum_{i=1}^{K} \bu_{i}\otimes f$, where $\bu_{i}$ is the $i$-th eigenvector of $\bL$ and $f = 2\frac{\sin(\omega_{0}t)}{(\omega_{0}t)}-\frac{\sin((\omega_{0}-2)t)}{((\omega_{0}-2)t)}+\frac{\sin((\omega_{0}-4)t)}{2((\omega_{0}-4)t)}$.
We set $K=2$ and $\omega_{0}=10$, and sample $M=300$ uniformly spaced time points over $t\in[-2.5,2.5]$ (see \cref{fig:syn_visual}). Visual inspection of the ground-truth heat map motivates placing the Gaussian-subspace atom center at $(v_{0},t_{0})=(36,0)$.

Given candidate diffusion scales $\calP_{v} = \set{1,0.5,0.1,0.01}$ and $\calP_{t} = \set{1,0.5,0.1,0.01}$, and the optimal pair $(\tau_{v}^{\star},\tau_{t}^{\star}) = (0.1, 0.01)$ is chosen via \cref{opt_diffusion_scales}. With this choice, we form $\bA^{\star}=\Pi_{\Sigma}\Pi_{\tau_{v}^{\star},\tau_{t}^{\star}}\Pi_{\Sigma}$, and the leading eigenvectors of $\bA^{\star}$ form an orthonormal basis of $\Sigma = (K,\omega_0)$-bandlimited signals that is maximally concentrated in the vertex-time domain. To assess noise robustness, we create the observed graph signal by adding Gaussian noise to the ground truth graph signal with different SNR levels ranging from $-8$ db to $2$ db. We compare the proposed basis with JFT, STVFT, and STVWT under the same graph bandwidth $K=2$ and (time) bandlimit $\omega_{0}=10$. For each method, we reconstruct by projecting the noisy observation onto that method’s span. As observed in \cref{fig:syn_RSE_SNR}, the Gaussian heat-kernel subspace basis achieves the lowest RSE across all tested SNRs, outperforming competing methods even at low SNR and remaining the top performer as SNR increases.


\begin{figure}[!htbp]
\centering
\vspace{-4mm}
\begin{subfigure}[t]{0.48\columnwidth}
\centering
\includegraphics[width=\columnwidth, trim={3.5cm 9cm 3.5cm 9cm}, clip]{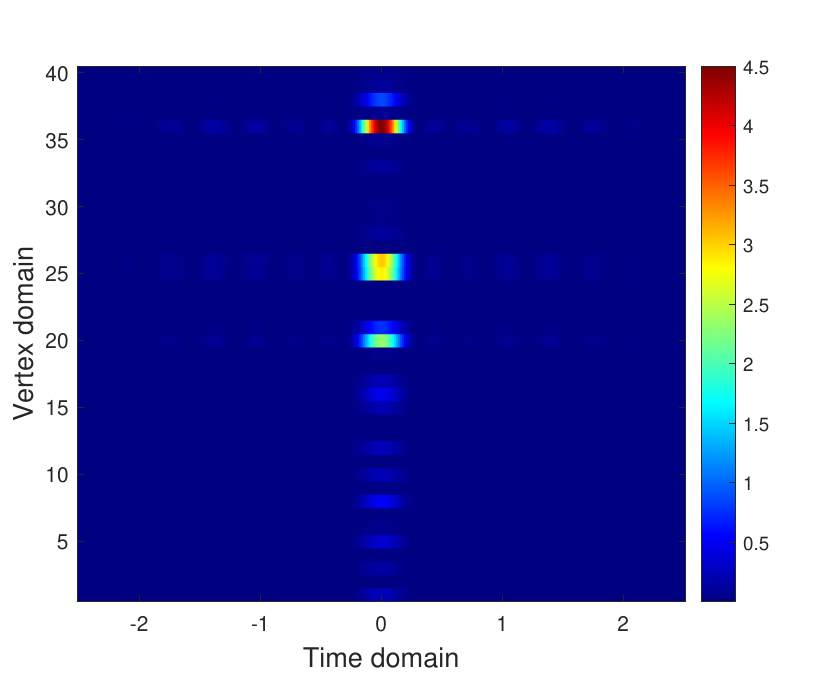}
\caption{Heat map of the power of the ground-truth graph signal.}
\label{fig:syn_visual}
\end{subfigure}
\hfill
\begin{subfigure}[t]{0.48\columnwidth}
\centering
\includegraphics[width=\columnwidth, trim={3.5cm 9cm 3.5cm 9cm}, clip]{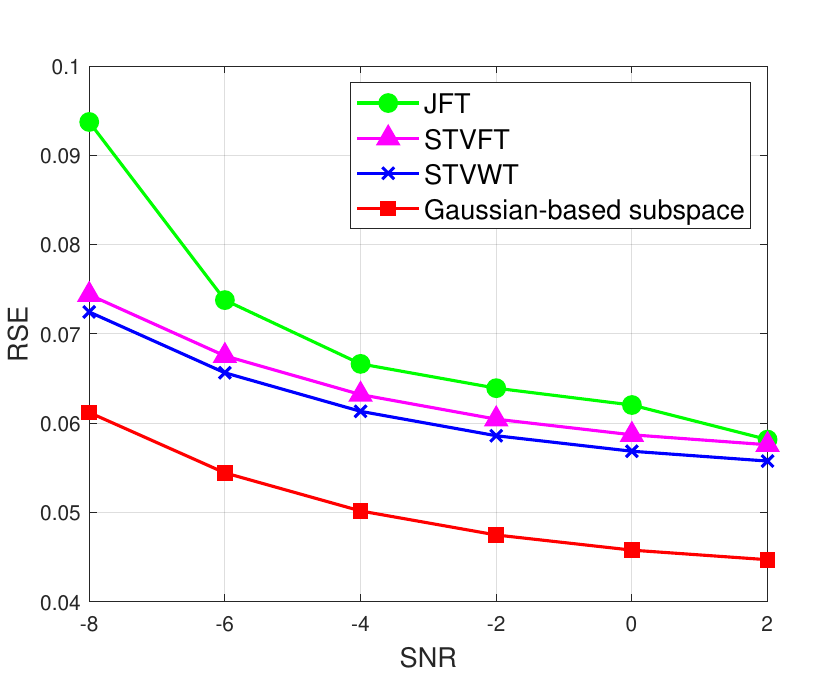}
\caption{RSE with varying SNR levels.}
\label{fig:syn_RSE_SNR}
\end{subfigure}
\caption{Ground-truth graph siganl and reconstruction error under varying SNR levels.}
\label{fig:Gaussian_MSE}
\end{figure}

\subsection{Graph topology inference results}

We examine the effectiveness of ECGL on a synthetic graph created using an Erd\H{o}s-R\`{e}nyi graph with $N=20$ nodes and a connection probability of $0.4$. The vertex-time graph signal is $(K,\omega_{0})$ bandlimited, and we construct it as $y = \sum_{i=1}^{K} \bu_{i}\otimes f$, where $\bu_{i}$ is the $i$-th eigenvector of $\bL$ and $f=\sin(\omega_{0}t)/(\omega_{0}t)$. We set $K=3$ and $\omega_{0}=10$ and evaluate the signal over $M=1000$ time instances. We then compare ECGL with two existing graph learning methods, Total Variation Graph Learning (TV-GL) \cite{stefania2019} and Estimated-Signal-Aided Graph Learning (ESA-GL) \cite{stefania2019}, using the performance metrics in \cref{gl_rho}, \cref{gl_adj}, and \cref{gl_pre_recall}. The results, summarized in \cref{tab:gl_results}, indicate that ECGL achieves the highest correlation coefficient with stronger alignment between the learned and true Laplacians. It also attains the lowest average recovery error, reflecting more accurate adjacency matrix estimation. Additionally, ECGL outperforms both TV-GL and ESA-GL in Precision and Recall, with a particularly notable improvement in Recall, which demonstrates its enhanced capability to recover true edges in the graph.

\begin{table}[!htb]
\caption{Performance comparison on graph topology reference. $\uparrow$ and $\downarrow$ indicate higher and lower values are better, respectively.}
\centering
\vspace{-2mm}
\resizebox{0.6\columnwidth}{!}{%
\begin{tabular}{lccc}
\toprule
& \text{TV-GL} & \text{ESA-GL} & \text{ECGL}  \\
\midrule
$\rho(\bL,\widetilde{\bL}) \uparrow$    &  0.9274      &  0.9267       &  0.9288      \\ 
$ \bar{\varepsilon}_{F} \downarrow$     &  0.0370      &  0.0368       &  0.0363       \\
\text{Precision }$\uparrow$              &  0.4338      &  0.4370       &  0.4500       \\
\text{Recall }$\uparrow$                 &  0.7284      &  0.7284       &  0.7780      \\
\bottomrule
\end{tabular}}

\label{tab:gl_results}
\end{table}

\section{Conclusion}\label{sec.conclu}

This paper establishes an uncertainty principle for vertex-time graph signals, unifying classical time-frequency and graph uncertainty principles. By defining vertex-time and spectral-frequency spreads, we quantify signal localization and identify maximally concentrated signals for constructing a vertex-time dictionary. This dictionary improves signal reconstruction, particularly under missing or intermittent data. Additionally, we introduce a graph topology inference method based on the uncertainty principle. Numerical experiments on synthetic and real datasets demonstrate superior reconstruction accuracy, noise robustness, and graph learning performance compared to existing methods. Future work will explore extensions to dynamic and large-scale graphs, further broadening the applicability of the proposed framework.

\begin{appendices}
    \section{Proof of \cref{UP_GGSP}}\label[Appendix]{sec:proof_Th3_UPGGSP}

Before proving the theorem, we present several technical preliminary results adapted from \cite{Tsitsvero2016} and \cite{Slepian1961}. The proofs essentially follow the same procedure as in \cite{Pollak1961} for continuous-time signals and are repeated here for completeness.

\begin{DefinitionA}
The angle between two functions $f$ and $g$ belonging to a Hilbert space is defined as
\begin{align}\label{angle_defi}
\theta(f,g)=\cos^{-1}\frac{\Re\{\angles{f,g}\}}{\norm{f}\norm{g}} 
\end{align}
where $\angles{\cdot,\cdot}$ and $\norm{}$ are the Hilbert space inner product and norm, respectively.  
\end{DefinitionA}



\begin{LemmaA}
\label{Lemma_the_min}
For a given function $f\in \ima(\Pi_{\SF})$, 
\begin{align*}
\inf_{g\in \ima(\Pi_{\VT})} \theta(f,g) = \cos^{-1} \frac{\norm{\Pi_{\VT} f}}{\norm{f}}
\end{align*}
is achieved by $g=k \cdot \Pi_{\VT}f$ for any constant $k>0$.
\end{LemmaA}
\begin{proof}
For any $g \in \ima(\Pi_{\VT})$, we have
\begin{align*}
\begin{aligned}
\Re\{\angles{f,g}\} \leq \abs{\angles{f,g}}&= \abs{\angles{f-\Pi_{\VT} f+\Pi_{\VT} f, g}} \\
&= \abs{\angles{\Pi_{\VT} f, g}} \leq \norm{\Pi_{\VT} f}\norm{g}
\end{aligned}
\end{align*}
where the second equality holds due to the fact that $f-\Pi_{\VT} f$ is orthogonal to $g$, and third equality is due to the Cauchy–Schwarz inequality. We thus have
\begin{align}\label{costheta}
\cos \theta(f,g)= \frac{\Re\{\angles{f,g}\}}{\norm{f}\norm{g}} \leq \frac{\norm{\Pi_{\VT} f}\norm{g}}{\norm{f}\norm{g}}=\frac{\norm{\Pi_{\VT} f}}{\norm{f}}.
\end{align}
Since $\cos \theta$ decreases monotonically in $[0,\pi]$ and equality in \cref{costheta} is achieved by choosing $g=k \cdot \Pi_{\VT}f$ with $k>0$, it follows that 
\begin{align*}
\inf_{g \in \ima(\Pi_{\VT})} \theta(f,g) = \cos^{-1} \frac{\norm{\Pi_{\VT} f}}{\norm{f}}.
\end{align*}
\end{proof}

We define the minimum angle between $\ima(\Pi_{\SF})$ and $\ima(\Pi_{\VT})$ as
\begin{align}\label{thetamin}
\theta_{\min}=\inf_{f \in \ima(\Pi_{\SF}), g\in \ima(\Pi_{\VT})}\theta(f,g).
\end{align}

\begin{PropositionA}[\cite{Colub1973}]
\label{min_angle}
The minimum angle $\theta_{\min}$ between $\ima(\Pi_{\SF})$ and $\ima(\Pi_{\VT})$ is given by 
\begin{align*}
\theta_{\min} = \cos^{-1}\sqrt{\lambda_{\max}\left(\Pi_{\SF}\Pi_{\VT}\Pi_{\SF}\right)},
\end{align*}
and is achieved by choosing $f=\psi_{0}$ and $g=\Pi_{\VT}\psi_{0}$ in \cref{thetamin}, where $\psi_{0}$ is an eigenvector of the Hermitian operator $\Pi_{\SF}\Pi_{\VT}\Pi_{\SF}$ corresponding to the eigenvalue $\lambda_{\max}(\Pi_{\SF}\Pi_{\VT}\Pi_{\SF})$.
\end{PropositionA}
\begin{proof}
From \cref{Lemma_the_min}, we have 
\begin{align*}
\begin{aligned}
&\inf_{f\in \ima(\Pi_{\SF}),~g\in \ima(\Pi_{\VT})}\theta(f,g)\\
&\qquad\qquad=\inf_{f\in \ima(\Pi_{\SF})}\cos^{-1}\frac{\norm{\Pi_{\VT} f}}{\norm{f}}\\
&\qquad\qquad= \cos^{-1}\sup_{f\in \ima(\Pi_{\SF})} \frac{\norm{\Pi_{\VT} f}}{\norm{f}},
\end{aligned}
\end{align*}
since $\cos(\cdot)$ decreases monotonically in $[0,\pi]$. For any $f\in \ima(\Pi_{\SF})$, $\Pi_{\SF} f=f$, and we have  
\begin{align}\label{pre_max}
\begin{aligned}
\frac{\norm{\Pi_{\VT} f}^{2}}{\norm{f}^{2}}&=\frac{\norm{\Pi_{\VT}\Pi_{\SF}f}^{2}}{\norm{f}^{2}} = \frac{\langle \Pi_{\VT}\Pi_{\SF} f, \Pi_{\VT}\Pi_{\SF} f\rangle}{\langle f,f\rangle} \\
&= \frac{\langle f, \Pi_{\SF}^{\ast}\Pi_{\VT}^{\ast}\Pi_{\VT}\Pi_{\SF} f\rangle}{\langle f,f\rangle}.
\end{aligned}
\end{align}
Since $\Pi_{\SF}$ and $\Pi_{\VT}$ are both Hermitian and idempotent operators, \cref{pre_max} yields
\begin{align*}
\frac{\norm{\Pi_{\VT} f}^{2}}{\norm{f}^{2}}= \frac{\angles{f, \Pi_{\SF}\Pi_{\VT}\Pi_{\SF} f}}{\angles{f,f}}.
\end{align*}
From the Rayleigh-Ritz theorem \cite{Horn1985}, we obtain
\begin{align}
\label{min_angle_RR_theorem}
\begin{aligned}
\sup_{f}\frac{\norm{\Pi_{\VT} f}^{2}}{\norm{f}^{2}}&=\sup_{f}\frac{\angles{f, \Pi_{\SF}\Pi_{\VT}\Pi_{\SF} f}}{\angles{f,f}} \\
&= \lambda_{\max}(\Pi_{\SF}\Pi_{\VT}\Pi_{\SF}),
\end{aligned}
\end{align}
and the proposition is proved. 
\end{proof}

We are now ready to prove \cref{UP_GGSP}.
Without loss of generality, we derive which values of $\beta_{\SF}$ are attainable for every choice of $\alpha_{\VT}$, assuming a unit-norm signal $f$. We consider two cases $\alpha_{\VT}=1$ and $\alpha_{\VT}\in (0,1)$. 

The case $\alpha_{\VT}=1$ means that the signal $f$ is supported only on the closure of $\calH_{\VT}$. From \cref{min_angle_RR_theorem} (by interchanging $\Pi_{\SF}$ and $\Pi_{\VT}$) and \cref{min_angle}, we obtain
\begin{align}\label{eq.sup_beta}
\sup_{f\in \ima(\Pi_{\VT}), \norm{f}=1} \beta_{\SF}^2 = \lambda_{\max}(\Pi_{\VT}\Pi_{\SF}\Pi_{\VT}).
\end{align}
Similarly, 
\begin{align}\label{eq.inf_beta}
\inf_{f\in \ima(\Pi_{\VT}), \norm{f}=1} \beta_{\SF}^2= 1-\lambda_{\max}(\Pi_{\VT}\overline{\Pi}_{\SF}\Pi_{\VT})
\end{align}
which is equivalent to the case of maximal concentration on the orthogonal complement subspace $\calH_{\SF}^{\perp}$. Thus, in the case of $\alpha_{\VT}=1$, $\beta_{\SF}^2$ lies in the interval $[1-\lambda_{\max}(\Pi_{\VT}\overline{\Pi}_{\SF}\Pi_{\VT}),\lambda_{\max}(\Pi_{\VT}\Pi_{\SF}\Pi_{\VT})]$.

Next, suppose $\alpha_{\VT} \in (0,1)$. We can decompose any signal $f$ as
\begin{align}\label{sig_decom}
f = \lambda \Pi_{\VT} f+ \gamma \Pi_{\SF}f+g
\end{align}
where $g$ is a signal orthogonal to both $\ima(\Pi_{\VT})$ and $\ima(\Pi_{\SF})$. Our goal is to find the nearest signal to $f$ in the space spanned by $\Pi_{\VT} f$ and $\Pi_{\SF} f$. To do this, we firstly calculate the inner products of \cref{sig_decom} successively with $f$, $\Pi_{\VT}f$, $\Pi_{\SF} f$ and $g$, and arrive at the system of equations:
\begin{align}\label{in_pro_sys_equ}
\centering
\left\{\begin{aligned}
&1 = \lambda \alpha_{\VT}^2 + \gamma \beta_{\SF}^2 + \angles{g,f},\\
&\alpha_{\VT}^2 = \lambda \alpha_{\VT}^2 + \gamma \angles{\Pi_{\SF} f, \Pi_{\VT} f},\\
&\beta_{\SF}^2 = \lambda \angles{\Pi_{\VT} f, \Pi_{\SF} f} + \gamma \beta_{\SF}^2,\\
&\angles{f,g} = \angles{g,g}.
\end{aligned}\right.
\end{align}
After eliminating $\angles{g,f}$, $\lambda$ and $\gamma$ from \cref{in_pro_sys_equ}, we obtain
\begin{align}\label{eliminating_results}
\begin{aligned}
&\beta_{\SF}^2-2\Re\{\angles{\Pi_{\VT} f,\Pi_{\SF} f}+\alpha_{\VT}^2 \\
&=\left( 1-\frac{\abs{\angles{\Pi_{\VT} f,\Pi_{\SF} f}}^2}{\alpha_{\VT}^2\beta_{\SF}^2}\right)
-\norm{g}^{2}\left( 1-\frac{\abs{\angles{\Pi_{\VT} f,\Pi_{\SF} f}}^2}{\alpha_{\VT}^2\beta_{\SF}^2}\right).
\end{aligned}
\end{align}

From \cref{min_angle}, the angle $\theta$ between $\Pi_{\VT} f\in \ima(\Pi_{\VT})$ and  $\Pi_{\SF} f \in \ima(\Pi_{\SF})$ satisfies
\begin{align}\label{angle_bound}
\theta \geq \cos^{-1}\sqrt{\lambda_{\max}\left(\Pi_{\SF}\Pi_{\VT}\Pi_{\SF}\right)}=\cos^{-1}\lambda_{\max}\left(\Pi_{\VT}\Pi_{\SF}\right).
\end{align}
Since $\norm{\Pi_{\VT} f}=\alpha_{\VT}$, $\norm{\Pi_{\SF} f}=\beta_{\SF}$, \cref{angle_defi} gives
\begin{align}
\begin{aligned}
\label{eq.bound_inequ}
\alpha_{\VT}\beta_{\SF} \cos\theta&= \Re\{\angles{\Pi_{\VT} f,\Pi_{\SF} f}\} \\
&\leq \abs{\angles{\Pi_{\VT} f,\Pi_{\SF} f}} \leq \alpha_{\VT} \beta_{\SF},
\end{aligned}
\end{align}
and we can further write
\begin{align}\label{angle_inequ}
0 \leq 1- \frac{\abs{\angles{\Pi_{\VT} f,\Pi_{\SF} f}}^2}{\alpha_{\VT}^2\beta_{\SF}^2} \leq 1-\cos^{2}\theta.
\end{align}

Completing the square on the left-hand side of \cref{eliminating_results}, and applying \cref{angle_inequ}, we obtain
\begin{align*}
\begin{aligned}
(\beta_{\SF}&-\alpha_{\VT}\cos\theta)^2\\
&= -\alpha_{\VT}^2\sin^2\theta+\left( 1-\frac{\abs{\angles{\Pi_{\VT} f,\Pi_{\SF} f}}^2}{\alpha_{\VT}^2\beta_{\SF}^2}\right)\\
&\qquad\qquad-\norm{g}^{2}\left( 1-\frac{\abs{\angles{\Pi_{\VT} f,\Pi_{\SF}f}}^2}{\alpha_{\VT}^2\beta_{\SF}^2}\right)\\
&\leq -\alpha_{\VT}^2\sin^2\theta+1-\cos^2\theta \\
& = (1-\alpha_{\VT}^2)\sin^2\theta = \sin^2 (\cos^{-1} \alpha_{\VT}) \sin^2\theta,
\end{aligned}
\end{align*}
where equality holds in the last inequality if and only if $\angles{\Pi_{\VT} f, \Pi_{\SF} f}$ is real and $g=0$ (a.e.). Therefore,
\begin{align*}
\beta_{\SF} \leq \cos(\theta-\cos^{-1}\alpha_{\VT}),
\end{align*}
from which it follows, by means of the bound \cref{angle_bound}, that
\begin{align*}
\beta_{\SF} \leq \cos(\cos^{-1}\sqrt{\lambda_{\max}\left(\Pi_{\SF}\Pi_{\VT}\Pi_{\SF}\right)}-\cos^{-1}\alpha_{\VT}).
\end{align*}
Thus,
\begin{align}\label{beta_range2}
\cos^{-1}\alpha_{\VT} + \cos^{-1}\beta_{\SF} \geq \cos^{-1}\sqrt{\lambda_{\max}(\Pi_{\SF}\Pi_{\VT}\Pi_{\SF})},
\end{align}
where equality is achieved by
\begin{align*}
f = p\psi_{0} + q \Pi_{\VT}\psi_{0}
\end{align*}
with 
\begin{align*}
p &= \sqrt{\frac{1-\alpha_{\VT}^2}{1-\lambda_{\max}(\Pi_{\SF}\Pi_{\VT}\Pi_{\SF})}}, \\
q &=\frac{\alpha_{\VT}}{\sqrt{\lambda_{\max}(\Pi_{\SF}\Pi_{\VT}\Pi_{\SF})}}-\sqrt{\frac{1-\alpha_{\VT}^2}{1-\lambda_{\max}(\Pi_{\SF}\Pi_{\VT}\Pi_{\SF})}}
\end{align*}
and $\psi_{0}$ is an eigenvector of the Hermitian operator $\Pi_{\SF}\Pi_{\VT}\Pi_{\SF}$ corresponding to the eigenvalue $\lambda_{\max}(\Pi_{\SF}\Pi_{\VT}\Pi_{\SF})$. 

Using $\norm{\overline{\Pi}_{\VT} f}= \sqrt{1-\alpha_{\VT}^2}$ and applying the same steps between \cref{sig_decom} and \cref{beta_range2} to the operators $\Pi_{\SF}\overline{\Pi}_{\VT}\Pi_{\SF}$, we obtain the inequality
\begin{align*}
\cos^{-1}\sqrt{1-\alpha_{\VT}^2} + \cos^{-1}\beta_{\SF} \geq \cos^{-1}\sqrt{\lambda_{\max}(\Pi_{\SF}\overline{\Pi}_{\VT}\Pi_{\SF})}.
\end{align*}
Similarly, with $\norm{\overline{\Pi}_{\SF} f}=\sqrt{1-\beta_{\SF}^2}$, we have
\begin{align*}
\cos^{-1}\alpha_{\VT} + \cos^{-1}\sqrt{1-\beta_{\SF}^2} \geq \cos^{-1}\sqrt{\lambda_{\max}(\overline{\Pi}_{\SF}\Pi_{\VT}\overline{\Pi}_{\SF})}
\end{align*}
and
\begin{align*}
\begin{aligned}
\cos^{-1}\sqrt{1-\alpha_{\VT}^2} &+ \cos^{-1}\sqrt{1-\beta_{\SF}^2} \\
&\geq \cos^{-1}\sqrt{\lambda_{\max}(\overline{\Pi}_{\SF}\overline{\Pi}_{\VT}\overline{\Pi}_{\SF})}.
\end{aligned}
\end{align*}

We have now verified the four inequalities in \cref{eq.Theorem_feasible region}. For $\beta_{\SF}=1$, $\alpha_{\VT}^2$ lies in the interval $[1-\lambda_{\max}(\Pi_{\SF}\overline{\Pi}_{\VT}\Pi_{\SF}),\lambda_{\max}(\Pi_{\SF}\Pi_{\VT}\Pi_{\SF})]$,
achieved by the eigenvectors of $\Pi_{\SF}\Pi_{\VT}\Pi_{\SF}$ that belong to $\Pi_{\SF}$ and their linear combinations. 
Proceeding similarly, one can show that all the values $\alpha_{\VT}$ and $\beta_{\SF}$ belonging to the border of $\Theta$ are achievable.  All the points inside $\Theta$ are achievable different combinations of the left and right singular vectors of $\Pi_{\SF}\Pi_{\VT}\Pi_{\SF}$,  $\Pi_{\SF}\overline{\Pi}_{\VT}\Pi_{\SF}$, $\overline{\Pi}_{\SF}\Pi_{\VT}\overline{\Pi}_{\SF}$ and $\overline{\Pi}_{\SF}\overline{\Pi}_{\VT}\overline{\Pi}_{\SF}$. This concludes the proof of \cref{UP_GGSP}.

\section{Proof of \cref{Lemma.joint_bounds}}\label[Appendix]{sec.Lemma.joint_bounds}

Using the fact that $\norm{f}^2_{L^{2}(\calV \times \calT)}=\sum_{v\in\calV}\norm{f(v,\cdot)}_{L^{2}(\calT)}^{2}=\int_{\calT}\norm{f(\cdot, t)}_{L^{2}(\calV)}^{2}\ud t$, we have
\begin{align}
\alpha_{\calS}^2
&= \frac{\norm{f}^2_{L^{2}(\calS)}}{\norm{f}^2_{L^{2}(\calV \times \calT)}}
= \frac{\sum_{v\in\calV'}\int_{\calT'}\abs{f(v,t)}^2 \ud t}{\sum_{v\in\calV}\int_{\calT}\abs{f(v,t)}^2 \ud t}\nn
& = \frac{\int_{\calT'}\norm{f(\cdot, t)}_{L^{2}(\calV')}^{2} \ud t}{\int_{\calT'}\norm{f(\cdot, t)}_{L^{2}(\calV)}^{2}\ud t} \frac{\sum_{v\in\calV}\norm{f(v,\cdot)}_{L^{2}(\calT')}^{2}}{\sum_{v\in\calV}\norm{f(v,\cdot)}_{L^{2}(\calT)}^{2}}\label{alpha_S}
\end{align}
and
\begin{align}
\beta_{\Sigma}^2
&=\frac{\norm{\calU f}_{L^{2}(\Sigma)}^{2}}{\norm{\calU f}^{2}_{L^{2}(\Lambda\times \Omega)}}
= \frac{\sum_{\lambda\in\Lambda'}\int_{\Omega'}\abs{\calU f(\lambda,\omega)}^2 \ud\omega}{\sum_{\lambda\in\Lambda}\int_{\Omega}\abs{\calU f(\lambda,\omega)}^2 \ud \omega} \nn
&= \frac{\int_{\Omega'}\norm{\calU f(\cdot,\omega)}_{L^{2}(\Lambda')}^{2} \ud\omega}{\int_{\Omega'}\norm{\calU f(\cdot,\omega)}_{L^{2}(\Lambda)}^{2}\ud\omega} \frac{\sum_{\lambda\in\Lambda}\norm{\calU f(\lambda,\cdot)}_{L^{2}(\Omega')}^{2}}{\sum_{\lambda\in\Lambda}\norm{\calU f(\lambda,\cdot)}_{L^{2}(\Omega)}^{2}}.\label{beta_sigma}
\end{align}
Since
\begin{align*}
\min_{v}\alpha^{2}_{\calT'}(v) \leq \frac{\sum_{v\in\calV}\norm{f(v,\cdot)}_{L^{2}(\calT')}^{2}}{\sum_{v\in\calV}\norm{f(v,\cdot)}_{L^{2}(\calT)}^{2}}\leq \max_{v}\alpha^{2}_{\calT'}(v)
\end{align*}
and
\begin{align*}
\inf_t \alpha^2_{\calV'}(t) \leq \frac{\int_{\calT'}\norm{f(\cdot, t)}_{\ell^{2}(\calV')}^{2}\ud t}{\int_{\calT'}\norm{f(\cdot, t)}_{\ell^{2}(\calV)}^{2}\ud t} \leq \sup_t \alpha^2_{\calV'}(t),
\end{align*}
the vertex-time spread $\alpha^2_{\calS}$ in \cref{alpha_S} is bounded by
\begin{align*}
\min_{v}\alpha^{2}_{\calT'}(v) \inf_t\alpha^2_{\calV'}(t)\leq \alpha^2_{\calS} \leq \max_{v}\alpha^{2}_{\calT'}(v) \sup_t\alpha^2_{\calV'}(t).
\end{align*} 
A similar argument bounds the spectral-frequency spread $\beta^2_{\Sigma}$ in \cref{beta_sigma} by
\begin{align*}
\min_{\lambda} \beta^{2}_{\Omega'}(\lambda)\inf_{\omega}\beta_{\Lambda'}^{2}(\omega)\leq \beta^{2}_{\Sigma}\leq  \max_{\lambda} \beta^{2}_{\Omega'}(\lambda)\sup_{\omega}\beta_{\Lambda'}^{2}(\omega),
\end{align*}
and the proof is complete.

\section{Proof of \cref{Theo.Feasible_Regoin_alpha}}\label[Appendix]{sec:proof_Th2_alpha}
Consider such two subsets $\calS=\calV'\times \calT'$ and $\Sigma = \Lambda'\times \Omega'$. In this case, the vertex-time limiting operator $\Pi_{\calS}$ and spectral-frequency limiting operator $\Pi_{\Sigma}$ could be written as $\Pi_{\calS}=\Pi_{\calV'}\otimes\Pi_{\calT'}$ and $\Pi_{\Sigma}=\Pi_{\Lambda'}\otimes\Pi_{\Omega'}$, respectively.
Based on the range of the spectral-frequency spread $\beta^2_{\Sigma}$ illustrated in \cref{eq.beta_sigma_range} as well as the monotonicity of $\cos^{-1}$, then the supremum of $\alpha_{\calS}$ achieves at
\begin{align*}
\begin{aligned}
\alpha_{\calS} \leq \cos&\Big(\cos^{-1}\sqrt{\lambda_{\max}\left(\Pi_{\Lambda'}\Pi_{\calV'}\Pi_{\Lambda'}\right)\lambda_{\max}\left(\Pi_{\Omega'}\Pi_{\calT'}\Pi_{\Omega'}\right)}\\
&\qquad\qquad-\cos^{-1}\left(\beta_{\Lambda'}^{\min}\beta_{\Omega'}^{\min}\right)\Big) 
\end{aligned}
\end{align*}
along with the property of tensor product $(\Pi_{\calV'}\otimes \Pi_{\calT'})(\Pi_{\Lambda'}\otimes \Pi_{\Omega'})(\Pi_{\calV'}\otimes \Pi_{\calT'})= (\Pi_{\calV'}\Pi_{\Lambda'}\Pi_{\calV'})\otimes (\Pi_{\calT'}\Pi_{\Omega'}\Pi_{\calT'})$ and $\lambda_{\max}((\Pi_{\calV'}\Pi_{\Lambda'}\Pi_{\calV'})\otimes (\Pi_{\calT'}\Pi_{\Omega'}\Pi_{\calT'}))=\lambda_{\max}(\Pi_{\calV'}\Pi_{\Lambda'}\Pi_{\calV'})\lambda_{\max}(\Pi_{\calT'}\Pi_{\Omega'}\Pi_{\calT'})$. Consider the fact of $\norm{\overline{\Pi}_{\Sigma}f}=\sqrt{1-\beta^2_{\Sigma}}$ with the range of $\beta^2_{\Sigma}$ shown in \cref{eq.beta_sigma_range} and apply the operator $\overline{\Pi}_{\Sigma}\Pi_{\calS}\overline{\Pi}_{\Sigma}$, then the supremum of $\alpha_{\calS}$ comes to
\begin{align*}
\alpha_{\calS} &\leq \cos\Big(\cos^{-1}\sqrt{\lambda_{\max}\left(\overline{\Pi}_{\Lambda'}\Pi_{\calV'}\overline{\Pi}_{\Lambda'}\right)\lambda_{\max}\left(\overline{\Pi}_{\Omega'}\Pi_{\calT'}\overline{\Pi}_{\Omega'}\right)}\\
&\qquad\qquad\qquad-\cos^{-1}\Big(\sqrt{1-(\beta_{\Omega'}^{\max}\beta^{\max}_{\Lambda'})^2}\Big)\Big).
\end{align*}
Similarly, when we apply the operators $\Pi_{\Sigma}\overline{\Pi}_{\calS}\Pi_{\Sigma}$ and $\overline{\Pi}_{\Sigma}\overline{\Pi}_{\calS}\overline{\Pi}_{\Sigma}$, other two inequalities are arrived at
\begin{align*}
\begin{aligned}
\alpha_{\calS} \geq &\sin\Big(\cos^{-1}\sqrt{\lambda_{\max}\left(\Pi_{\Lambda'}\overline{\Pi}_{\calV'}\Pi_{\Lambda'}\right)\lambda_{\max}\left(\Pi_{\Omega'}\overline{\Pi}_{\calT'}\Pi_{\Omega'}\right)}\\
&\qquad\qquad-\cos^{-1}\left(\beta_{\Lambda'}^{\min} \beta_{\Omega'}^{\min}\right)\Big) 
\end{aligned}
\end{align*}
and
\begin{align*}
\begin{aligned}
\alpha_{\calS} \geq &\sin\Big(\cos^{-1}\sqrt{\lambda_{\max}\left(\overline{\Pi}_{\Lambda'}\overline{\Pi}_{\calV'}\overline{\Pi}_{\Lambda'}\right)\lambda_{\max}\left(\overline{\Pi}_{\Omega'}\overline{\Pi}_{\calT'}\overline{\Pi}_{\Omega'}\right)}\\
&\qquad\qquad-\cos^{-1}\Big(\sqrt{1-\left(\beta_{\Lambda'}^{\max}\beta^{\max}_{\Omega'}\right)^2}\Big)\Big).
\end{aligned}
\end{align*}
This completes the proof of \cref{Theo.Feasible_Regoin_alpha}.

\section{Proof of \cref{Theo:localized_both_domains}}\label[Appendix]{sec:proof_the3}

In the forward direction, by repeatedly applying \cref{eq.vertex_time_limited_signals} and \cref{eq.spec_fre_limited_signals}, we obtain
\begin{align*}
\Pi_{\SF}\Pi_{\VT}\Pi_{\SF} f = \Pi_{\SF}\Pi_{\VT} f = \Pi_{\VT}\Pi_{\SF} f = \Pi_{\SF} f = f.
\end{align*}
Conversely, without loss of generality, suppose $f$ is nonzero and
\begin{align}\label{eq.eigenvector_BDB}
\Pi_{\SF}\Pi_{\VT}\Pi_{\SF} f = f.
\end{align}
Multiplying by $\Pi_{\SF}$ and using $\left(\Pi_{\SF}\right)^2=\Pi_{\SF}$, we obtain 
\begin{align}\label{eq.eigenvector_BDB_Bf}
\Pi_{\SF}\Pi_{\VT}\Pi_{\SF} f = \Pi_{\SF} f.
\end{align}
By equating \cref{eq.eigenvector_BDB} and \cref{eq.eigenvector_BDB_Bf}, we have
\begin{align*}\
\Pi_{\SF} f =f
\end{align*}
indicating that $f$ is perfectly localized in the spectral-frequency domain. Applying the Rayleigh-Ritz theorem \cite{Horn1985} and the Hermitian property of $\Pi_{\SF}$ (i.e., $\left(\Pi_{\SF}\right)^{\ast}=\Pi_{\SF}$), we derive
\begin{align*}
1=\max_{f} \frac{f^{\ast}\Pi_{\SF}\Pi_{\VT}\Pi_{\SF} f}{f^{\ast}f} = \max_{f}\frac{f^{\ast}\Pi_{\VT} f}{f^{\ast}f},
\end{align*}
which confirms \cref{eq.vertex_time_limited_signals}, establishing perfect localization in the vertex-time domain. This completes the proof of \cref{Theo:localized_both_domains}.

\section{Proof of \cref{Theo.solu_optimization}} \label[Appendix]{sec:proof_theo_solution}
Substituting the joint bandlimiting constraint within the objective function in \cref{eq.optimization_problem}, we obtain
\begin{align}\label{eq.optimization_problem_v2}
\begin{aligned}
\xi_{i} = \argmax_{\xi_{i}}~& \norm{\Pi_{\VT}\Pi_{\Sigma}\xi_{i}}\\
\ST & \norm{\xi_{i}} = 1, ~\angles{\xi_{i},\xi_{j}}=0,~j\neq i. 
\end{aligned}
\end{align}
By applying the Rayleigh-Ritz theorem \cite{Horn1985}, the solutions to \cref{eq.optimization_problem_v2} are the eigenvectors of the operator $(\Pi_{\VT}\Pi_{\Sigma})^{\ast}\Pi_{\VT}\Pi_{\Sigma}=\Pi_{\Sigma}\Pi_{\VT}\Pi_{\Sigma}$, i.e., $\Pi_{\Sigma}\Pi_{\VT}\Pi_{\Sigma} \xi_{i} = \lambda_{i} \xi_{i}$, where $\lambda_{i}$ are eigenvalues corresponding to the eigenvectors $\xi_{i}$ \cite{Royden2017}. Given that $\Pi_{\Sigma}\xi_{i}=\xi_{i}$ and $\left(\Pi_{\Sigma}\right)^{\ast}=\Pi_{\Sigma}$, we have
\begin{align}
\begin{aligned}
    \lambda_{j}\delta_{ij} &= \angles{\xi_{i},\Pi_{\Sigma}\Pi_{\VT}\Pi_{\Sigma}\xi_{j}} \\
    &= \angles{\left(\Pi_{\Sigma}\right)^{\ast}\xi_{i},\Pi_{\VT}\Pi_{\Sigma}\xi_{j}} =  \angles{\xi_{i},\Pi_{\VT}\xi_{j}}.
\end{aligned}
\end{align}
This completes the proof of \cref{Theo.solu_optimization}.
\end{appendices}


\end{document}